\documentclass[a4paper,UKenglish,cleveref, autoref, thm-restate,nolineno]{socg-lipics-v2021}

\usepackage{amsthm,amsmath,amssymb,amsfonts}
 \usepackage{pgfmath}
\usepackage{dsfont}
\usepackage{xspace}
\usepackage{graphicx}
\usepackage{nicefrac}
 \usepackage{colortbl}
\usepackage{subcaption} 
 \usepackage{xcolor} 

\usepackage{tikz}
\usetikzlibrary{arrows.meta,calc,positioning}
\usepackage{boxedminipage}
\usepackage{framed}
\usepackage{tabularx}
\usepackage{tcolorbox}
\usepackage{etoolbox}
\usepackage{xifthen}
\usepackage{listings}

\newcommand{\shortversion}[1]{}

\newcommand{\mcm}[3]{\newcommand{#1}[#2]{{\ensuremath{#3}}}} 
\mcm{\Nbb}{0}{\mathbb{N}}
\mcm{\Zbb}{0}{\mathbb{Z}}
\mcm{\Rbb}{0}{\mathbb{R}}
\mcm{\Cbb}{0}{\mathbb{C}}
\mcm{\Qbb}{0}{\mathbb{Q}}
\mcm{\Acal}{0}{\cal A}
\mcm{\Bcal}{0}{\cal B}
\mcm{\Ccal}{0}{\cal C}
\mcm{\Dcal}{0}{\cal D}
\mcm{\Ecal}{0}{\cal E}
\mcm{\Fcal}{0}{\cal F}
\mcm{\Gcal}{0}{\cal G}
\mcm{\Hcal}{0}{\cal H}
\mcm{\Ical}{0}{\cal I}
\mcm{\Jcal}{0}{\cal J}
\mcm{\Kcal}{0}{\cal K}
\mcm{\Lcal}{0}{\cal L}
\mcm{\Mcal}{0}{\cal M}
\mcm{\Ncal}{0}{\cal N}
\mcm{\Ocal}{0}{{\cal O}}
\mcm{\Pcal}{0}{{\cal P}}
\mcm{\Qcal}{0}{{\cal Q}}
\mcm{\Rcal}{0}{{\cal R}}
\mcm{\Scal}{0}{{\cal S}}
\mcm{\Tcal}{0}{{\cal T}}
\mcm{\Ucal}{0}{{\cal U}}
\mcm{\Vcal}{0}{{\cal V}}
\mcm{\Wcal}{0}{{\cal W}}
\mcm{\Xcal}{0}{{\cal X}}
\mcm{\Ycal}{0}{{\cal Y}}
\mcm{\Zcal}{0}{{\cal Z}}
\newcommand{\wtilde}[1]{\widetilde{#1}}
\newcommand{\Oh}{\mathcal{O}}

\newcommand{\bigoh}[0]{{\mathcal O}}

\newcommand{\dist}{\operatorname{\rho}}

\newcommand{\edm}{\textsc{EDM}\xspace}

\newcommand{\dedmcomp}{$d$-\textsc{\edm Completion}\xspace}
\newcommand{\dedmcompshort}{$d$-\textsc{EDMC}\xspace}

\DeclareMathOperator{\operatorClassNP}{NP}
\newcommand{\classNP}{\ensuremath{\operatorClassNP}\xspace}

\DeclareMathOperator{\operatorClassFPT}{FPT\xspace}
\newcommand{\classFPT}{\ensuremath{\operatorClassFPT}\xspace}
\DeclareMathOperator{\operatorClassW}{W}
\newcommand{\classW}[1]{\ensuremath{\operatorClassW[#1]}}
\DeclareMathOperator{\operatorClassParaNP}{Para-NP\xspace}

\DeclareMathOperator{\operatorClassXP}{XP\xspace}
\newcommand{\classXP}{\ensuremath{\operatorClassXP}\xspace}

\newcommand{\classParaNPComplete}{\ensuremath{\operatorClassParaNP}\xspace-Complete}
   
\newlength{\RoundedBoxWidth}
\newsavebox{\GrayRoundedBox}
\newenvironment{GrayBox}[1]%
   {\setlength{\RoundedBoxWidth}{.93\textwidth}
    \def\boxheading{#1}
    \begin{lrbox}{\GrayRoundedBox}
       \begin{minipage}{\RoundedBoxWidth}}%
   {   \end{minipage}
    \end{lrbox}
    \begin{center}
    \begin{tikzpicture}%
       \node(Text)[draw=black!20,fill=white,rounded corners,%
             inner sep=2ex,text width=\RoundedBoxWidth]%
             {\usebox{\GrayRoundedBox}};
        \coordinate(x) at (current bounding box.north west);
        \node [draw=white,rectangle,inner sep=3pt,anchor=north west,fill=white] 
        at ($(x)+(6pt,.75em)$) {\boxheading};
    \end{tikzpicture}
    \end{center}}
    
\newenvironment{defproblemx}[2][]{\noindent\ignorespaces%
                                \FrameSep=6pt%
                                \parindent=0pt%
                \vspace*{-1.5em}
                \ifthenelse{\isempty{#1}}{%
                  \begin{GrayBox}{\textsc{#2}}%                
                }{%
                  \begin{GrayBox}{\textsc{#2}  parameterized by~{#1}}%  
                }
                \begin{tabular*}{\textwidth}{@{\hspace{.1em}} >{\itshape} p{1.8cm} p{0.8\textwidth} @{}}%        
            }{
                \end{tabular*}%
                \end{GrayBox}%
                \ignorespacesafterend
            }

\usepackage[framemethod=tikz]{mdframed}

\definecolor{mycolor}{rgb}{0.122, 0.435, 0.698}
\newmdenv[innerlinewidth=0.02pt, roundcorner=4pt,linecolor=mycolor,innerleftmargin=6pt,
innerrightmargin=6pt,innertopmargin=6pt,innerbottommargin=6pt]{mybox}
\newmdenv[innerlinewidth=0.5pt, roundcorner=4pt,linecolor=black,innerleftmargin=6pt,
innerrightmargin=6pt,innertopmargin=6pt,innerbottommargin=6pt]{myboxblack}
\newmdenv[innerlinewidth=0.5pt, roundcorner=4pt,linecolor=mycolor,innerleftmargin=6pt,
innerrightmargin=6pt,innertopmargin=6pt,innerbottommargin=6pt]{myboxthick}

\pdfoutput=1 %uncomment to ensure pdflatex processing (mandatatory e.g. to submit to arXiv)
\hideLIPIcs  %uncomment to remove references to LIPIcs series (logo, DOI, ...), e.g. when preparing a pre-final version to be uploaded to arXiv or another public repository

\title{Algorithms for Euclidean Distance Matrix Completion: Exploiting Proximity to Triviality}

\titlerunning{Algorithms for Euclidean Distance Matrix Completion}

\author{Fedor V. Fomin}{University of Bergen, Norway}{fedor.fomin@uib.no}{https://orcid.org/0000-0003-1955-4612}{Supported by the Research Council of Norway under BWCA project (grant no.~314528) and by the European Research Council (ERC) under the European Union's Horizon 2020 research and innovation programme (NewPC grant agreement No.  101199930)
    }

\author{Petr A. Golovach}{University of Bergen, Norway}{petr.golovach@uib.no}{https://orcid.org/0000-0002-2619-2990}{Supported by the Research Council of Norway under the BWCA (grant no.~314528) and  Extreme-Algorithms (grant no.~355137) projects.}

\author{M. S. Ramanujan}{University of Warwick, UK}{r.maadapuzhi-sridharan@warwick.ac.uk}
{https://orcid.org/0000-0002-2116-6048}{Supported by Engineering and Physical Sciences Research Council (EPSRC) grant EP/V044621/1.}

\author{Saket Saurabh}{Institute of Mathematical Sciences, Chennai, India and University of Bergen, Norway}{saket@imsc.res.in}{https://orcid.org/0000-0001-7847-6402}{Supported by the European Research Council (ERC) under the European Union's Horizon 2020 research and innovation programme (grant agreement No. 819416); and Swarnajayanti
Fellowship grant DST/SJF/MSA-01/2017-18.  \begin{minipage}{0.2\textwidth}\includegraphics[width=0.9\textwidth]{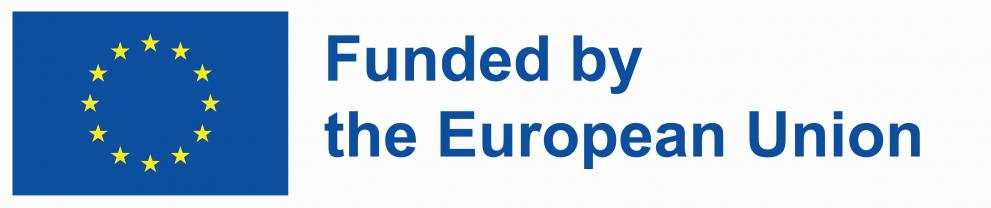}\end{minipage}}

\authorrunning{F. V. Fomin, P. A. Golovach, M. S. Ramanujan, S. Saurabh}

\Copyright{Fedor V. Fomin, Petr A. Golovach, M. S. Ramanujan, Saket Saurabh} 


\begin{document}

\maketitle

\begin{abstract}
In the $d$-Euclidean Distance Matrix Completion ($d$-EDMC) problem, one aims to determine whether a given partial matrix of pairwise distances can be extended to a full Euclidean distance matrix in $d$ dimensions. This problem is a cornerstone of computational geometry with numerous applications. While classical work on  this problem often focuses on exploiting connections to semidefinite programming typically leading to approximation algorithms, we focus on exact algorithms and propose a novel {\em distance-from-triviality parameterization} framework to obtain tractability results for $d$-EDMC. 

We identify key structural patterns in the input that capture entry density, including chordal substructures and coverability of specified entries by fully specified principal submatrices. We obtain:

\begin{enumerate}
    \item The first fixed-parameter algorithm (FPT algorithm) for {\dedmcompshort} parameterized by $d$ and the maximum number of unspecified entries per row/column. This is achieved through a novel {\em compression algorithm} that reduces a given instance to a  submatrix on $\mathcal{O}(1)$ rows (for fixed values of the parameters).
    \item The first FPT algorithm for {\dedmcompshort}  parameterized by $d$ and the minimum number of fully specified principal submatrices whose entries cover all specified entries of the given matrix. This result is also achieved through a compression algorithm. 
\item A polynomial-time algorithm for {\dedmcompshort} when both $d$ and the minimum fill-in of a  natural graph representing the specified entries are fixed constants.
This result is achieved by combining tools from distance geometry and algorithms from real algebraic geometry.
\end{enumerate}
Our work identifies interesting parallels between EDM completion and graph problems, with our algorithms exploiting techniques from both domains.
\end{abstract}

\section{Introduction}

Consider the following problem: given an \( n \times n \) symmetric hollow matrix \( M \), and integer $d$, determine whether \( M \) is a Euclidean Distance Matrix (EDM)  that is {isometrically embeddable} into $\mathbb{R}^d$, also said to be 
%. We also use the term 
 \emph{realizable in $\mathbb{R}^d$}. In other words,  determine whether there exists a set \( P = \{p_1, \dots, p_n\} \) of \( n \) points in \( \mathbb{R}^d \) such that
%\[
$M_{ij} = \| p_i - p_j \|_2^2$, 
%\]
where \( M_{ij} \) denotes the \( (i,j) \)-th entry of \( M \), and \( \| p_i - p_j \|_2 \) is the Euclidean distance between \( p_i \) and \( p_j \). This question of determining the existence of such a point configuration \( P \) (up to congruence) is a central problem in Distance Geometry, whose solution is well understood and dates back to the works of Cayley~\cite{Cayley1841} and Menger~\cite{Menger1928}.

In this paper, we are interested in the version of this problem where the matrix \( M \) is {\em partial}, i.e., some entries may be ``unspecified.''
In this version, called \dedmcomp{} (\dedmcompshort{}), the goal is to decide whether the missing entries in a matrix can be filled in so that the result is an EDM that can be realized in \( \mathbb{R}^d \). In general, such questions of reconstructing point configurations from {\em partial} distance measurements naturally occur across numerous disciplines.  The study of \dedmcompshort{} in particular, goes back to mathematicians and mechanical engineers of the nineteenth century---see, for example, the work of Cauchy from 1813~\cite{cauchy1813polygones}---and is at the heart of rigidity theory.
This problem has many applications in different areas, such as molecular conformation in bioinformatics~\cite{Havel1983,Crippen1988,Crippen1991}, dimensionality reduction in machine learning and statistics, and localization in wireless sensor networks~\cite{Alfakih1999,Doherty2001,Biswas2004}.

Saxe \cite{saxe1979embeddability} and Yemini \cite{Yemini1979} showed that the \dedmcompshort{} problem is already \emph{strongly} NP-hard in dimensions \(d=1\) and \(d=2\).
While numerous heuristics for solving \dedmcompshort{} under specific structural assumptions on~\(M\) have been proposed in the literature \cite{fang2012euclidean,Weinberger2004,Biswas2006,Candes2010},
\emph{theoretical} results with provable performance guarantees are scarce. This lack of rigorous algorithmic understanding of the problem is a primary motivation for us to explore the computational complexity of this problem and in particular, its parameterized complexity.

In this work we adopt the \emph{distance-from-triviality} viewpoint to study the complexity of \dedmcompshort{}.  Distance-from-triviality is  a
well-established paradigm in parameterized complexity, where the idea is to introduce a parameter (or a combination of parameters) that
measures how far a given instance lies from a class that can be solved
efficiently; algorithms then exploit this “distance’’ to achieve tractability.

\smallskip
\noindent\emph{What is the notion of “triviality’’ for \dedmcompshort?}  
Two natural extremal cases suggest themselves.

\noindent{\bf 1. }\emph{All entries are  unspecified}. 
 If no entry of the matrix is given, \emph{any} Euclidean distance matrix of the appropriate size is a valid completion, so the instance is trivial.  
        Unfortunately, even a very small deviation from this case already yields hardness: Saxe's reduction~\cite{saxe1979embeddability} shows NP-hardness when each row and column contains \emph{exactly two} specified entries.  
        Consequently, parameterizing by distance to the completely unspecified matrix appears unpromising.
        
\noindent{\bf 2. }\emph{All entries are specified}.  
        At the opposite extreme, when every entry is given, one need only test whether the matrix is an EDM that embeds isometrically into $\mathbb{R}^{d}$, a task solvable in polynomial time. 
        As we show in this paper, distance to this notion of ``triviality''  is a much more interesting parameter.

We may therefore formulate our main question as follows.

\begin{tcolorbox}[
   colback = blue!6,           % pastel background
   colframe = blue!65!black,   % darker border
   coltitle = black,           % title text colour
  % title = Guiding Question,   % box title
   boxrule = 1pt,              % border thickness
   arc = 3pt,                  % rounded corners
   left = 2mm, right = 2mm,    % interior padding
   top = 1mm, bottom = 1mm,
   fonttitle = \bfseries       % bold title
]
\centering
\parbox{0.93\linewidth}{%
Which notions of distance to a fully specified matrix can be exploited
algorithmically for \dedmcompshort{}?
}
\end{tcolorbox}
From the standpoint of algorithms with provable guarantees, this question is largely unexplored.  
A notable exception is the work of Berger, Kleinberg, and Leighton~\cite{BergerKL99}.  
They show that if, in every row of an $n\times n$ partial matrix $M$, at least $3n/4+1$ entries are specified, then one can decide in polynomial time whether $M$ can be completed to a $3$-EDM.  
There is, however, a crucial caveat: their theorem assumes that the target point set is in \emph{general position} and that  no ten points lie on a quadric surface.  
We make no such assumptions---the points to be embedded need not be in general position or even distinct---rendering the problem substantially more realistic but also  more challenging.

 \subsection{Our results and methods}

 We begin by observing that high density of specified entries (or, equivalently, a numerical sparsity of unspecified entries) does not by itself make the problem any easier. 
To exploit sparsity algorithmically, one must therefore identify and leverage additional {\em structural properties} that reflect the density of entries.

\begin{restatable}{theorem}{denseObs}
\label{obs:dense}
 $\forall\varepsilon\in(0,1)$, \dedmcompshort{} remains strongly \classNP-hard even for instances $(M,d)$ with $d=2$ in which the $n\times n$ matrix $M$ contains at most $\varepsilon n$ unspecified entries.
\end{restatable}

The proof of the theorem is an adaptation of the classic complexity result of  Saxe~\cite{saxe1979embeddability}.   
It also underscores that the general-position assumption employed by Berger, Kleinberg, and Leighton \cite{BergerKL99} is essential for their polynomial-time algorithm.

\subsubsection{Distance from Triviality I}
Our first main algorithmic result  shows that imposing even a mild structural constraint on the unspecified entries---beyond mere numerical  sparsity---makes the problem tractable.
Our notion of density of a matrix is expressed by {\em forbidden $t$-block patterns}.

\begin{definition}[$t$-block pattern]
{\em 	For an $n\times n$ matrix $M$ and disjoint sets $I_1,I_2\subseteq [n]$, we define by $M[I_1,I_2]$ the submatrix of $M$ indexed by the rows in $I_1$ and the columns in $I_2$. We say that $M[I_1,I_2]$ is a {\em $t$-block pattern} if it is a $t\times t$ submatrix of $M$ where every entry is unspecified.  We say that $M$ {\em excludes} a $t$-block pattern if there is no $I_1$ and $I_2$ such that $M[I_1,I_2]$ is a $t$-block pattern. See 
%\Cref{fig:partial-edm-3-block}.
top of \Cref{fig:partial-edm-3-block-with-graphc}.
}
\end{definition}
 
%\usetikzlibrary{calc,positioning,arrows.meta}
\tikzset{
  vertex/.style     ={circle, draw, minimum size=7mm, inner sep=0pt},
  normaledge/.style ={line width=0.7pt},
  rededge/.style    ={line width=1.3pt, red}
}
\newcommand{\Edge}[3][]{%
  \pgfmathtruncatemacro{\A}{#2}%
  \pgfmathtruncatemacro{\B}{#3}%
  \draw[#1] (v\A) -- (v\B);
}

\begin{figure}[ht]
\begin{nolinenumbers}
\centering
\begin{subfigure}{0.42\textwidth}
\centering
\[
\left[
\begin{array}{ccccccccc}
0 & \ast & \ast & \ast & \ast & 2 & \ast & \ast & \ast \\
\ast & 0 & 1 & \ast & \ast & \ast & \ast & \ast & \ast \\
\cellcolor{red!30}\ast & 1 &  0 & \ast & \cellcolor{red!30}\ast & \ast & \cellcolor{red!30}\ast &  2 &  \ast \\
\ast & \ast & \ast & 0 & 1 & \ast & \ast & \ast & 2 \\
\ast & \ast & \ast & 1 & 0 & 1 & \ast & \ast & \ast \\
2 & \ast & \ast & \ast & 1 & 0 & 1 & \ast & \ast \\
\ast & \ast & \ast & \ast & \ast & 1 & 0 & \ast & \ast \\
\cellcolor{red!30}\ast & \ast & 2 & \ast & \cellcolor{red!30}\ast & \ast & \cellcolor{red!30}\ast &  0 &  \ast \\
\cellcolor{red!30}\ast & \ast &  \ast & 2 & \cellcolor{red!30}\ast & \ast & \cellcolor{red!30}\ast & \ast &  0
\end{array}
\right]
\]
\phantomcaption	
\end{subfigure}

\begin{subfigure}{0.44\textwidth}
\centering
\begin{tikzpicture}[scale=0.6]
  \def\R{3.2}
  \foreach \i/\ang in {1/90,2/50,3/10,4/-30,5/-70,6/-110,7/-150,8/-190,9/-230}
    \node[vertex] (v\i) at ({\R*cos(\ang)},{\R*sin(\ang)}) {\i};

  % edges from numeric (non-*) entries
  \Edge[normaledge]{1}{6}
  \Edge[normaledge]{2}{3}
  \Edge[normaledge]{3}{8}
  \Edge[normaledge]{4}{5}
  \Edge[normaledge]{4}{9}
  \Edge[normaledge]{5}{6}
  \Edge[normaledge]{6}{7}
\end{tikzpicture}
\phantomcaption
\label{fig:partial-edm-3-blocka}
\end{subfigure}
\hfill
% (b) RIGHT: graph from * entries (unchanged)
\begin{subfigure}{0.44\textwidth}
\centering
\begin{tikzpicture}[scale=0.6]
  \def\R{3.2}
  \foreach \i/\ang in {1/90,2/50,3/10,4/-30,5/-70,6/-110,7/-150,8/-190,9/-230}
    \node[vertex] (v\i) at ({\R*cos(\ang)},{\R*sin(\ang)}) {\i};

  % black edges: * entries (excluding highlighted ones)
  \Edge[normaledge]{1}{2}
  \Edge[normaledge]{1}{4}
  \Edge[normaledge]{1}{5}
  \Edge[normaledge]{1}{7}
  \Edge[normaledge]{2}{4}
  \Edge[normaledge]{2}{5}
  \Edge[normaledge]{2}{6}
  \Edge[normaledge]{2}{7}
  \Edge[normaledge]{2}{8}
  \Edge[normaledge]{2}{9}
  \Edge[normaledge]{3}{4}
  \Edge[normaledge]{3}{6}
  \Edge[normaledge]{3}{9}
  \Edge[normaledge]{4}{6}
  \Edge[normaledge]{4}{7}
  \Edge[normaledge]{4}{8}
  \Edge[normaledge]{5}{7}
  \Edge[normaledge]{6}{8}
  \Edge[normaledge]{6}{9}
  \Edge[normaledge]{8}{9}

  % red edges: highlighted star entries
  \Edge[rededge]{1}{3}
  \Edge[rededge]{3}{5}
  \Edge[rededge]{3}{7}
  \Edge[rededge]{1}{8}
  \Edge[rededge]{5}{8}
  \Edge[rededge]{7}{8}
  \Edge[rededge]{1}{9}
  \Edge[rededge]{5}{9}
  \Edge[rededge]{7}{9}
\end{tikzpicture}
\phantomcaption
\label{fig:graph-from-matrixb}
\end{subfigure}

\caption{Top: A \( 9 \times 9 \) partial matrix $M$ that  excludes $4$-block pattern. It contains a $3$-block pattern formed by rows \( \{3, 8, 9\} \) and columns \( \{1, 5, 7\} \). Bottom Left: the graph $G$ underlying the partial EDM matrix $M$.
% from \Cref{fig:partial-edm-3-block}.  
 The adjacencies of $G$ are formed by the unspecified entries in the matrix. Bottom Right: the complement graph  $\overline{G}$  induced by the unspecified entries of $M$. The  red edges of $\overline{G}$  form $K_{3,3}$, however $\overline{G}$ does not contain $K_{4,4}$.}
\label{fig:partial-edm-3-block-with-graphc}
\end{nolinenumbers}
\end{figure}
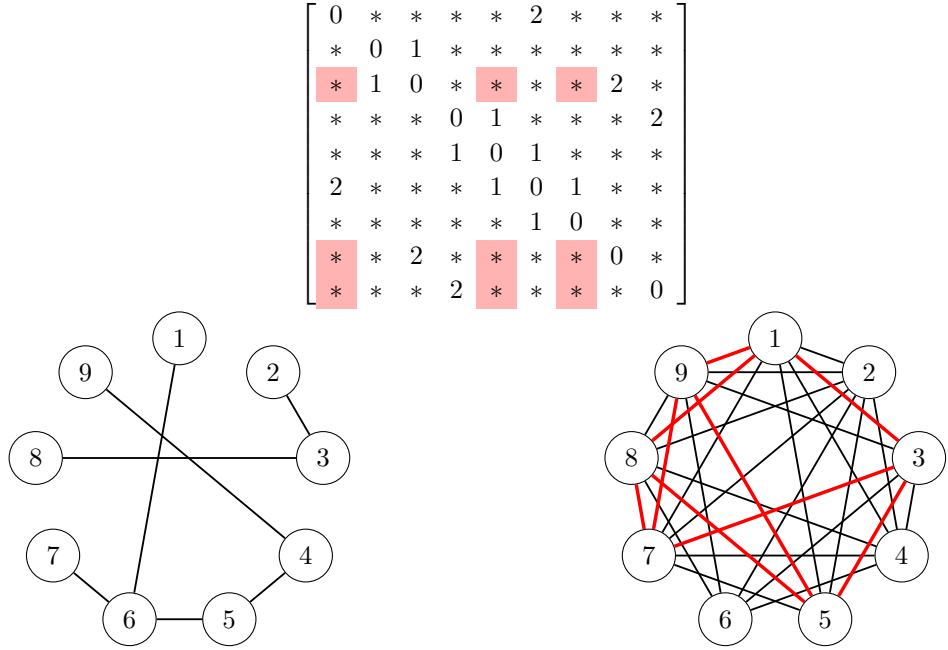

\begin{example}
Consider a partial matrix $M$. 
$M$ excludes a 1-block pattern if and only if it is fully specified;
  if $M$ has at most $\Delta$ unspecified entries in total, then $M$ excludes a $\lfloor\sqrt{\Delta+1}\rfloor$-block pattern;
  if each row of $M$ contains at most $\Delta$ unspecified entries, then $M$ excludes a $(\Delta+1)$-block pattern.
\end{example}

We are now ready to state our first tractability result. By {\em solving} an instance of {\dedmcompshort}, we mean deciding\footnote{As is standard in Computational Geometry, all our algorithms operate under the \emph{real RAM} computational model, assuming that basic operations over real numbers can be executed in unit time.} whether the answer is yes or no. We also say that an instance $(M,d)$  is {\em equivalent} to an instance $(M',d')$ if solving both instances leads to the same answer. 
 
\begin{restatable}[Compression for $K_{t,t}$-free  complements]{theorem}{compressionThm}
\label{thm:compression}
There is a polynomial-time algorithm that, given an instance
\((M,d)\) of \dedmcompshort{} in which \(M\) excludes a \(t\)-block pattern
for some \(t \ge 1\),
either 
%(i) 
solves the problem, or 
%(ii)
 outputs an equivalent
instance \((M',d)\), such that $M'$ is a principal submatrix of $M$ with  
$(d+1)^{\mathcal{O}(t^{2})}$ rows and columns. 
\end{restatable}

The algorithm in \Cref{thm:compression} compresses the input instance in polynomial time to an
equivalent matrix whose size is independent of~$n$.  This compression, however, does {not} guarantee that the numerical entries of the reduced
matrix depend only on the parameters $t$ and~$d$; hence it is {not} a
kernel in the strict parameterized complexity sense~ \cite{cygan2015parameterized,kernelizationbook19}. 

Moreover, 
\Cref{thm:compression} by itself does not yield a complete algorithm for solving \dedmcompshort{} and we need 
 the following algorithm, which relies on  tools from real algebraic geometry.

\begin{restatable}{theorem}{exactAlgoCompletion}
\label{thm:exactAlgoCompletion}
There is an algorithm that, given an $n\times n$ partial matrix $M$ and $d\in\mathbb{N}$, runs in time $2^{\mathcal{O}(n^{2}\log n)}$ and correctly decides whether $(M,d)$ is a yes-instance of \dedmcompshort{}.
\end{restatable}

 Now, by 
combining \Cref{thm:compression} with \Cref{thm:exactAlgoCompletion}, we obtain the following  \classFPT algorithm for {\dedmcompshort} with parameters $d$ and $t$.

\begin{restatable}{corollary}{fptKttCor}
\label{cor:fptKtt}
For every $d,t\in\mathbb{N}$, and $n\times n$ partial matrix \(M\) excluding a \(t\)-block pattern, 
\dedmcompshort{} is solvable  in time $2^{(d+1)^{\bigoh(t^2)}}+n^{\bigoh(1)}$. 
\end{restatable}

\begin{remark}[Generality of \Cref{thm:compression} and \Cref{cor:fptKtt}]
{
 A common way to capture the structure of a partial matrix, see e.g.  \cite{alfakih2018euclidean,BergerKL99}, is to encode it by an underlying  graph $G$ as follows.  
Let $G$ be a graph with vertex set  $V(G)=\{1,\ldots,n\}$.
 An $n \times n$ symmetric hollow partial matrix $M = (m_{ij})$  is said to be a \emph{$G$-partial} matrix if the entry $m_{ij}$ is specified\footnote{In this paper, the specified entries of $M$ will always be in $\mathbb{R}_{\geq 0}$.} if and only if $ij \in E(G)$. We refer to $G$ as the \emph{underlying graph} of~$M$, see \Cref{fig:partial-edm-3-block-with-graphc}.
By studying partial matrices through a graph-theoretic lens motivated by this definition, we can identify direct correspondences between their properties and those of their underlying graphs. The condition of \Cref{obs:dense} also translates naturally to the property that the average vertex degree of $\overline{G}$ is at most $\varepsilon$. 
  For instance, saying that a partial matrix $M$  {excludes a $t$-block pattern} is equivalent to saying that the complement $\overline{G}$ of its underlying graph $G$ does not contain the  complete bipartite graph $K_{t,t}$ as a subgraph.

   Although \Cref{obs:dense} shows that \dedmcompshort{} remains intractable even on matrices where the underlying graph's complement has  constant average degree, \Cref{thm:compression} and \Cref{cor:fptKtt} demonstrate that imposing $K_{t,t}$-freeness on the complement of the underlying graph (i.e., excluding a  $t$-block pattern in the input matrix) radically alters the complexity landscape and leads to tractability for a large class of instances.

Many widely studied graph classes with numerous applications can be expressed as $K_{t,t}$-free graphs for a suitable $t$. In particular, planar graphs are $K_{3,3}$-free by Kuratowski's theorem, graphs excluding a fixed minor $H$ are $K_{t,t}$-free for some $t$ depending on $H$, graphs of bounded expansion and nowhere dense graphs are $K_{t,t}$-free for an appropriate $t$, and any graph of maximum degree $\Delta$ (or more generally, bounded degeneracy) is $K_{t,t}$-free for $t=\Delta+1$ or $t$ equal to one plus the degeneracy.

}	
\end{remark}

In the rest of the paper, we freely switch between the graph and matrix viewpoints when describing structural properties of partial matrices.

Notice that in the more restricted case when every row of $M$ contains at most $\Delta$ unspecified entries, the maximum vertex degree of $\overline{G}$ is at most $\Delta$. In this case, we show that the size of the compression provided by \Cref{thm:compression} can be refined as follows. 
\begin{restatable}
%[{\rm{$\star$}}]
%[Compression for degree-bounded complements]
{theorem}{DeltaCompr}
\label{thm:Deltacompr}
There is a polynomial-time algorithm that, given an instance \((M,d)\) of \dedmcompshort{} such that  every row of \(M\) has at most \(\Delta\) unspecified entries,  either
%(i) 
solves the problem, or
%  (ii)
   outputs an equivalent instance 
   \((M',d)\), such that $M'$ is a principal $n'\times n'$ submatrix of $M$, where     \(n'\leq (d+1)(\Delta+1)^{2}\).
\end{restatable}

Combined with \Cref{thm:exactAlgoCompletion}, we have the following algorithm.
\begin{restatable}{corollary}{fptDeltacor}
\label{cor:fptDeltacor}
For every $d,\Delta\in\mathbb{N}$, and $n\times n$ partial matrix \(M\) such that
every row of \(M\) has at most \(\Delta\) unspecified entries, \dedmcompshort{} is solvable  in time $2^{\bigoh(d^2\cdot \Delta^4\cdot (\log d+\log \Delta))}+n^{\bigoh(1)}$. 
\end{restatable}

\subsubsection{Distance from Triviality II}

Consider the case where  all specified entries of $M$ lie in one
{fully specified principal submatrix} $M'$.
In this case, \dedmcompshort{} is polynomial-time solvable: one first determines
the minimum dimension $d$ for which $M'$ is isometrically embeddable in
$\mathbb{R}^{d}$ and then assigns values to the remaining entries of $M$ so
that the new points lie at appropriate distances.
So, a natural definition of distance to this notion of triviality is the least number  of fully specified principal submatrices such that their entries cover all specified entries of $M$. In terms of graphs, if all specified entries of $M$ lie in one
{fully specified principal submatrix}, then the underlying graph $G$ consists of a clique and isolated vertices. Thus all edges of $G$ are covered by one clique. So, when we say that the  specified entries in $M$ can be covered by $k$ fully specified principal submatrices, this is equivalent to saying that the underlying graph $G$ admits an \emph{edge clique cover} of size~$k$, that is,
the edges of $G$ can be covered by at most $k$ cliques.

\begin{remark}
%[Incomparability of edge clique cover size and $K_{t,t}$-freeness]
{  A small edge clique cover does \emph{not} imply that
$\overline{G}$ is $K_{t,t}$-free for any fixed~$t$.
Indeed, the disjoint union of two copies of $K_{n}$ can be covered by just two cliques, yet its complement contains $K_{n,n}$.
Thus, the edge clique cover parameter is \emph{incomparable} with the assumption of $K_{t,t}$-freeness, which motivates
%This motivates 
the parameterization by edge clique cover.
% and leads to our next result.
}
\end{remark}

\begin{restatable}
%[Compression for edge clique cover]
{theorem}{CoverCompr}
\label{thm:Covercompr}
There is a polynomial-time algorithm that, given an instance \((M,d)\) of \dedmcompshort{} in which \(M\) is an $n\times n$ \(G\)-partial matrix and the graph \(G\) is given together with an edge clique cover \(\mathcal{C}\) of size \(k\),
%(i) 
 either solves the problem, or
%  (ii)
   outputs an equivalent instance 
  \((M',d)\), such that $M'$ is an $n'\times n'$ principal submatrix of $M$, where  $n'\le (d+1)k^{2}$.
  
\end{restatable}

\begin{restatable}[{\rm{$\star$}}]{corollary}{fptECC}
\label{cor:fptECC}
For every $d,k\in\mathbb{N}$, and $n\times n$ partial matrix  \(M\) such that
the specified entries in \(M\) can be covered by at most \(k\) fully specified principal submatrices, \dedmcompshort{} is solvable  in time $2^{\bigoh(d^2\cdot k^4\cdot (\log d+\log k))}+2^{2^{\bigoh(k)}}\cdot n^{\bigoh(1)}$. 
\end{restatable}

Let us next give a brief overview of the proof technique behind  \Cref{thm:compression}, \Cref{thm:Deltacompr} and \Cref{thm:Covercompr}. Say that an instance $(M,d)$ of {\dedmcompshort} is {\em efficiently $\alpha$-reducible} if there is a polynomial-time computable pair comprising a  fully specified $\alpha\times \alpha$ principal submatrix indexed by $X\subseteq [n]$ and an element $w\in X$ such that $(M,d)$ is equivalent to $(M-w,d)$. Here, $M-w$ is defined as the matrix obtained from $M$ by  removing the row and column indexed by $w$.  The element $w$ is called an {\em irrelevant} element.
Clearly as long as such a pair $X$ and $w$ can be found, we can iteratively reduce the instance by deleting $w$. To prove these three theorems, we show the following, where $G$ is the underlying graph of $M$. 
\begin{itemize}
	\item if $\overline G$ is $K_{t,t}$-free, then either the instance  is solvable in polynomial time or it is already compressed to the claimed size or it is efficiently $d^{\bigoh(t)}$-reducible; and 
	\item if $\overline G$ has max-degree at most $\Delta$, then either the instance is solvable in polynomial time or it is already compressed to the claimed size or it is efficiently $\bigoh(d\cdot \Delta)$-reducible; and 
	\item given an edge clique cover of size $k$ for $G$, either the instance is solvable in polynomial time or it is already compressed to the claimed size or it is efficiently $\bigoh(d\cdot k)$-reducible.
\end{itemize}

Proving each of the above statements is done along similar lines, but there are certain instance-specific aspects that we exploit in each case. In fact, for the first two statements, we show the existence of an $\alpha$ depending only on the parameters such that the index set of {\em any} fully specified $\alpha\times \alpha$ principal submatrix must contain an irrelevant element. Then, we can use the sparsity of $\overline G$ to infer the existence of such a fully specified principal submatrix of $M$. For the third statement, we observe that one of the cliques in the given edge clique cover of $G$ must have size at least $\alpha$ (where $\alpha$ depends on the parameters) and then we argue that such a clique contains an irrelevant element.

\subsubsection{Distance from Triviality III}

Our final result concerns a further class of matrices that admit polynomial-time
completion algorithms.  
Recall that a graph $G$ is \emph{chordal} if it has no induced cycles of
length greater than $3$, see \Cref{fig:matrix-and-chordal-graph}.
Bakonyi and Johnson~\cite{Bakonyi1995} proved the following: when $G$ is
chordal, a $G$-partial matrix $M$ can be completed to an EDM that is realizable 
in $\mathbb{R}^{d}$  {if and only if} the submatrix induced by the vertex set of each maximal
clique of $G$ is itself an EDM realizable in $\mathbb{R}^{d}$.
Since an $n$-vertex chordal graph has at most $n$ maximal
cliques and these can be listed in polynomial time~\cite{Golumbic80}, Laurent~\cite{Laurent2001} observed that this
characterization yields a straightforward polynomial-time algorithm for
\dedmcompshort{} as follows. For each maximal clique of $G$ one checks, in polynomial time, whether the
corresponding principal submatrix of $M$ (all of whose entries are specified, by definition) is an EDM realizable in $\mathbb{R}^{d}$. So, considering partial matrices with chordal underlying graphs as our notion of triviality, let us address the next notion of distance from triviality.

\tikzset{
  vertex/.style = {circle, draw, fill=white, minimum size=6pt, inner sep=0pt},
  edge/.style   = {line width=0.8pt},
}

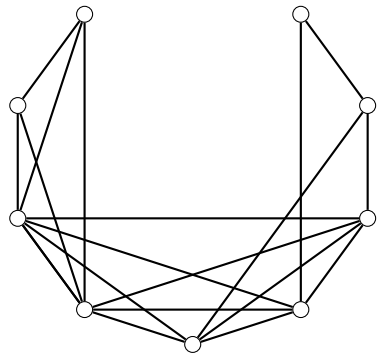
\begin{figure}[ht]
\begin{nolinenumbers}
\centering

% --- (a) Matrix on the left ---
\begin{subfigure}{0.47\textwidth}
\centering
\renewcommand{\arraystretch}{1.1}
\[
\left[
\begin{array}{ccccccccc}
0 & 1 & 1 & 1 & * & * & * & * & * \\      % 2
1 & 0 & 1 & 1 & * & * & * & * & * \\      % 3
1 & 1 & 0 & 1 & 1 & 1 & 1 & * & * \\      % 4
1 & 1 & 1 & 0 & 1 & 1 & 1 & * & * \\      % 5
* & * & 1 & 1 & 0 & 1 & 1 & 1 & * \\      % 6
* & * & 1 & 1 & 1 & 0 & 1 & * & 1 \\      % 7
* & * & 1 & 1 & 1 & 1 & 0 & 1 & * \\      % 8
* & * & * & * & 1 & * & 1 & 0 & 1 \\      % 9
* & * & * & * & * & 1 & * & 1 & 0        % 10
\end{array}
\right]
\]
\phantomcaption
\label{fig:edm-chordal}
\end{subfigure}
\hfill
% --- (b) Graph on the right ---
\begin{subfigure}{0.47\textwidth}
\centering
\begin{tikzpicture}[scale=1.1, rotate=90]
  % vertices on a circle (2..10)
  \foreach \i in {2,...,10} {
    \node[vertex] (v\i) at ({360/10 * (\i - 1)}:2.2) {};
  }
  % clique 1: v2–v5
  \foreach \a/\b in {2/3,2/4,2/5,3/4,3/5,4/5}
    \draw[edge] (v\a)--(v\b);
  % clique 2: v4–v8
  \foreach \a/\b in {4/5,4/6,4/7,4/8,5/6,5/7,5/8,6/7,6/8,7/8}
    \draw[edge] (v\a)--(v\b);
  % extra edges
  \draw[edge] (v8)--(v9);
  \draw[edge] (v9)--(v10);
  \draw[edge] (v7)--(v10);
  \draw[edge] (v6)--(v9);
\end{tikzpicture}
\phantomcaption
\label{fig:graph-chordal}
\end{subfigure}

\caption{A partial EDM matrix whose underlying graph is chordal.}
\label{fig:matrix-and-chordal-graph}
\end{nolinenumbers}
\end{figure}

The popular measure of the distance of a graph $G$ to a chordal graph is the \emph{fill-in} of a graph, which is the minimum number of edges that should be added to $G$ to make it chordal. Our next theorem extends the polynomial-time algorithm of 
Bakonyi and Johnson~\cite{Bakonyi1995}  for chordal graphs to graphs with constant fill-in. 

\begin{restatable}
%[XP algorithm parameterized by $d$ and fill-in]
{theorem}{XPFillIn}
\label{thm:XPAlgorithmbyFillIn}
Let $M$ be an $n\times n$ $G$-partial matrix and let $d\in\mathbb{N}$.  
There is an algorithm that runs in  
$
  d^{\mathcal{O}(kd)}\; 2^{\mathcal{O}(k^{2})}\; n^{\mathcal{O}(k)}
$
time, where $k$ is the size of a minimum fill-in of $G$, and correctly decides whether $(M,d)$ is a yes-instance of \dedmcompshort{}.
\end{restatable}

Laurent~\cite{Laurent2001} presented an XP algorithm, that is, polynomial for
every fixed $k$, which decides whether a $G$-partial matrix can be completed to
an EDM realizable in \emph{some} Euclidean space, where $k$ is the
minimum size of a fill-in of $G$.  This algorithm, however, does not determine the
\emph{smallest} dimension $d$ that suffices, and therefore does not 
solve \dedmcompshort{}.
\Cref{thm:XPAlgorithmbyFillIn} can thus be viewed as a non-trivial extension
of Laurent's result, tailored specifically for embeddability into a prescribed
dimension~$d$.

We next give a brief outline of the techniques behind this algorithm. 
Due to the result of Bakonyi and Johnson~\cite{Bakonyi1995}, it is necessary and sufficient to decide whether a given $G$-partial matrix $M$ can be extended to a $G'$-partial matrix $M'$ satisfying the following conditions:
(i) $G'$ is chordal, (ii) $G'$ is a supergraph of $G$, and (iii) for every maximal clique of $G'$, the submatrix of $M'$ induced by the clique's vertex set is an EDM realizable in $\mathbb{R}^{d}$.
For us, $G'$ is defined by a fill-in set $X$, that is, $G' = G + X$, where $X$ is a set of non-edges of $G$ whose addition to $G$ makes it chordal. Thus, the task reduces to computing $M'$ by assigning values to the entries corresponding to the edges in $X$, ensuring that condition~(iii) is satisfied. To do so, we reduce this task to testing the truth of an existentially quantified statement over a bounded set of polynomial equations and inequalities and then finally invoke an algorithm of Basu, Pollack, and Roy~\cite{BasuPR06}. We note that a similar approach was used in~\cite{BentertFGRS25}, however our case is more complex as the maximal cliques of $G'$ may share edges, requiring us to handle multiple EDMs simultaneously in this reduction.

 \subsection{Related Work}

The study of \emph{Euclidean distance matrices} (EDMs) dates back to the pioneering works of Cayley~\cite{Cayley1841} and Menger~\cite{Menger1928}.  
Foundational results were later obtained by Schoenberg~\cite{Schoenberg1935} and by Young and Householder~\cite{Young1938}.  
The subject was further developed in a series of papers by Gower, Critchley, Farebrother, and others~\cite{Gower1982,Gower1985,Critchley1988,Farebrother1987}.  
Schoenberg~\cite{Schoenberg1935} also established the deep connection between EDMs and positive-semidefinite (PSD) matrices.  
A comprehensive treatment of the vast literature on EDMs lies beyond the scope of this article; we refer the interested reader to the monographs and surveys~\cite{Blumenthal1970,Crippen1988,deza1997geometry,Laurent1998,Dattorro2008,alfakih2018euclidean}.

The \textsc{EDM Completion} problem, that is the problem of deciding whether an input $G$-partial EDM  could be completed to an EDM (without constraints on the rank of the resulting matrix) could be formulated as a semidefinite programming problem (SDP) and thus can be solved approximately in polynomial time, up to any given accuracy. 
The complexity of the \textsc{EDM Completion} problem is unknown, and its membership in \classNP{} remains open. Several works have investigated how structural properties of the graph $G$ influence the complexity of \textsc{EDM Completion}.

As already mentioned, Bakonyi and Johnson~\cite{Bakonyi1995} proved that when $G$ is \emph{chordal}, a $G$-partial EDM admits a completion to a full EDM if and only if the submatrix indexed by every maximal clique of $G$ is itself an EDM. Laurent~\cite{Laurent2001} exploited this characterization to design a polynomial-time algorithm for \textsc{EDM Completion} 
%(the problem of deciding of whether an input $G$-partial EDM  could be completed to EDM) 
when $G$ is chordal. Moreover, she extended this result to graphs that can be made chordal by adding only a constant number of edges.   In other words, she provided an XP algorithm parameterized by the \emph{fill-in} of $G$. 
This result for chordal graphs in the EDM setting is a natural extension of the analogous result for the positive semidefinite (PSD) matrix completion problem
~\cite{Grone1984}.

Both \textsc{EDM Completion} and \dedmcomp{} have numerous applications in
statistics~\cite{deleeuw1982mds}, chemistry~\cite{Crippen1988}, sensor
network localization~\cite{Biswas2006}, and many other areas; see the survey
of Laurent~\cite{laurent2008matrix} for additional references.
Since the rank function is not convex, \dedmcompshort{} cannot be formulated as a semidefinite program. It is \classNP-complete for \(d=1\) and remains \classNP-hard for every fixed \(d \ge 2\), even when restricted to partial matrices with specified entries in \(\{1,2\}\)~\cite{saxe1979embeddability,Yemini1979}. Numerous heuristics have been proposed in the literature (e.g., \cite{Crippen1988,glunt1998molecular,havel1991evaluation,hendrickson1995molecule,kuntz1993distance,more1997global,Weinberger2004}), but to our knowledge there are no algorithms for \dedmcompshort{} with provable performance guarantees. A notable exception is the work of Berger, Kleinberg, and Leighton \cite{BergerKL99}, which addresses \dedmcompshort{} for \(d=3\); however, their general-position assumptions differ substantially from those in this paper. Schaefer in \cite{schaefer2013realizability} proved that  \dedmcompshort (and actually 2-EDMC) 
 is complete for the existential theory of the reals. See also  the $\exists \mathbb{R}$ compendium   \cite{schaefer2024existential}, problem CG62 (Graph Realizability).
The area closely related to \textsc{EDM Completion} is rigidity theory. Rigidity theory, dating back to mathematicians and mechanical
engineers of the nineteenth century~\cite{cauchy1813polygones}, studies when a system of points connected by fixed-length bars is rigid. 
Here, the basic problem is to determine when a partial distance matrix uniquely determines a set of points in Euclidean space up to congruence. 
See \cite{asimow1978rigidity,connelly1991generic,graver1993combinatorial} for an introduction to this area.

\section{Preliminaries}

For every $x\in {\mathbb N}_{>0}$, we use $[x]$ to denote the set $\{1,\dots,x\}$.

\paragraph*{Distance spaces and matrices}  
Let $X$ be a set. A function $\dist \colon X\times X \to \mathbb{R}_{\ge 0}$  is a \emph{distance} on $X$ if: 
% \begin{itemize}
% \item[(i)]
(i)~$\dist$ is symmetric, that is,  for any $x, y \in X$, $\dist(x, y) = \dist(y, x)$, and  
% \item[(ii)] 
(ii)~$\dist(x, x) = 0$ for all  $x\in X$.
% \end{itemize}
Then, $ (X, \dist)$ is called a \emph{distance space}. 
%More precisely, recall that 
%for two points %(vectors)  
%$p, q \in \mathbb{R}^d$, the Euclidean distance between $p$ and $q$ is  $\| p - q \|_2=\sqrt{\langle %p,p\rangle+ \langle q,q\rangle - 2\langle p,q\rangle}$. 
A distance space $(X,\dist)$ for a finite $X=\{x_1,\ldots,x_n\}$ can be equivalently defined by the  \emph{distance matrix}, in which the value in row $i$ and column $j$
is $\dist_{ij}^2$, where 
%, denoted by $\dist_{ij}$ is defined as
%\[
%D(\rho)=\left(
% \begin{matrix}
%0 & \dist_{1,2}^2 & \dist_{1,3}^2 & \dots & \dist_{1,n}^2 \\
%\dist_{2,1}^2 & 0 & \dist_{2,3}^2 & \dots & \dist_{2,n}^2 \\
%\dist_{3,1}^2 & \dist_{3,2}^2 & 0 & \dots & \dist_{3,n}^2 \\
%\vdots & \vdots & \vdots & \ddots & \vdots \\
%\dist_{n,1}^2 & \dist_{n,2}^2 & \dist_{n,3}^2 & \dots & 0 \\
%\end{matrix} \right)
%\]
%where 
$\dist_{ij}=\dist(x_i,x_j)$ for $i,j\in\{1,\ldots,n\}$.
Throughout the paper we do not distinguish metric spaces and the corresponding distance matrices.

A distance space $(X,\dist)$ is \emph{isometrically embeddable} into $\mathbb{R}^d$ if there is a map, called \emph{isometric embedding}, $\varphi \colon X\to \mathbb{R}^d$ such that 
$\dist(x,y)=\|\varphi(x)-\varphi(y)\|_2$ for all $x,y\in X$, where
$\| p - q \|_2=\sqrt{\langle p,p\rangle+ \langle q,q\rangle - 2\langle p,q\rangle}$ for $p,q\in\mathbb{R}^d$.
Notice that we do not require $\varphi$ to be injective, that is, several points of $(X,\dist)$ may be mapped to the same point of $\mathbb{R}^d$.
Throughout the paper, whenever we mention an embedding, we mean an isometric embedding. Moreover, when we use the term $d$-embedding, we are referring to embedding into $\mathbb{R}^{d}$.  A $d$-embeddable  
distance space is {\em strongly $d$-embeddable} if it is not $(d-1)$-embeddable. %We use $X^{(2)}$ to denote the set of unordered pairs of two distinct elements of $X$. 
As convention, we assume that the empty set of points is $d$-embeddable for every $d$.
A symmetric $n\times n$ matrix $D=(d_{ij})$ over $\mathbb{R}_{\geq 0}$ with $d_{ii}=0$ for all $i\in\{1,\ldots,n\}$ is a \emph{Euclidean distance matrix} (\edm) in $\mathbb{R}^d$ if $D=D(\dist)$ is the distance matrix of a distance space $(X,\dist)$ embeddable in $\mathbb{R}^d$. 
Suppose that $\Xcal=(X,\dist)$ is embeddable into $\mathbb{R}^d$. The ordered set 
$P=(p_1,\ldots,p_n)$ of points in $\mathbb{R}^d$ is said to be a \emph{realization} of $\Xcal$ if there is an embedding $\varphi \colon X\to \mathbb{R}^d$ such that $\varphi(x_i)=p_i$ for all $i\in\{1,\ldots,n\}$.

For a $d$-embeddable distance space $(X,\dist)$, a set $Y\subseteq X$ is a \emph{metric basis} if, given an isometric embedding $\varphi$ of $(Y,\dist)$ into $\mathbb{R}^d$, there is a unique way to extend $\varphi$ to an isometric embedding of $(X,\dist)$. Equivalently, if a 
realization of $(Y,\dist)$ is fixed then the embedding of any point of $X\setminus Y$ in a $d$-embedding of $(X,\dist)$ is unique.  
In our paper we also use the embeddability characterization based on the properties of Cayley-Menger matrices. Towards this, we use terminology and notation from \cite{BentertFGRS25}. 
For $r+1$ points $x_0,x_1, \dots, x_r$ of distance space $(X, \dist)$
the \emph{Cayley-Menger determinant} $CM ( x_0,x_1, \dots, x_r)$ is the determinant of the matrix obtained from the distance matrix induced by these points by prepending a row and a column whose first element is zero and the other elements are one.

\begin{proposition}[{\cite[Chapter~IV]{Blumenthal1970}}]\label{thm:Blumental}
A distance space $\Xcal= (X, \rho)$ with $n$ points is strongly 
% embeddable into $\mathbb{R}^d$ 
$d$-embeddable if and only if there exist $d+1$ points, say $X_d = \{ x_0, \ldots, x_d \}$, such that:
\begin{enumerate}
\item  $(-1)^{j+1}CM ( x_0,x_1, \dots, x_j)>0$ for $1\leq j\leq d$, and 
\item for any $x, y \in X\setminus X_d$,
\[CM ( x_0,x_1, \dots, x_d,x) =CM ( x_0,x_1, \dots, x_d,y)=CM ( x_0,x_1, \dots, x_d,x,y)=0.\] 
\end{enumerate}
Equivalently (see, for example,  \cite{SidiropoulosWW17}), 
$\Xcal$ is strongly 
% embeddable into $\mathbb{R}^d$  
$d$-embeddable if and only if there is a set of $d+1$ points $X_d=\{x_0,\ldots,x_d\}$ such that 
% $(X_d, \rho)$ is strongly embeddable into $\mathbb{R}^d$ and 
$(\{x_0,\ldots,x_j\},\dist)$ is strongly $j$-embeddable 
% into $\mathbb{R}^j$  
for all $j\in\{1,\ldots, d\}$, and
for every $x,y\in X\setminus X_d$, $(X_d\cup\{x\}\cup \{y\})$ is $d$-embeddable.
% into  $\mathbb{R}^d$.
\end{proposition}

\begin{proposition}[{\cite[Chapter~IV]{Blumenthal1970}}]\label{prop:basis}
 Let $(X,\dist)$ be a $d$-embeddable distance space, and let  $B\subseteq X$ be a metric basis. 
 Then an embedding of $B$ in $\mathbb{R}^d$ is unique up to rigid transformations, and given an   embedding $\varphi$ of $(B,\dist)$ into $\mathbb{R}^d$, there is a unique way to extend $\varphi$ to an embedding of $(X,\dist)$. Equivalently, if a 
realization of $(B,\dist)$ is fixed then the embedding of any point of~$X\setminus B$ in a $d$-embedding of $(X,\dist)$ 
is unique.  
\end{proposition}

Let $(X,\dist)$ be a  
distance space. For a nonnegative integer $r$, 
we say that $Y\subseteq X$ of size $r+1$ is \emph{independent} if $(Y,\dist)$ is strongly $r$-embeddable.

\begin{proposition}[{\cite[Chapter~IV]{Blumenthal1970}}]\label{prop:matroid}
Let $(X,\dist)$ be a strongly $d$-embeddable distance space. Then the following hold: {\em (i)} any single-element set is independent, {\em (ii)} if $Y \subseteq X$ is independent, then any~$\emptyset \subset Z \subseteq Y$ is independent, and {\em (iii)} if $Y,Z \subseteq X$ are independent and $|Y|>|Z|$ then there is a~$y\in Y\setminus Z$ such that $Z\cup\{y\}$ is independent, 

{\em (iv)} the maximum size of an independent set is $d+1$ and any independent set~$Y$ of size $d+1$ is a metric basis.
 \end{proposition}

\begin{proposition}[{\cite[Chapter~IV]{Blumenthal1970}} and \cite{AlencarBLL15,SipplS85}]\label{prop:realization}
Given a distance space $\Xcal=(X,\dist)$ with $n$ points and a positive integer $d$, in $\Oh(n^3)$ time, it can be decided whether $\Xcal$ can be embedded into $\mathbb{R}^d$ and, if such an embedding exists, then a metric basis and a realization can be constructed in this running time.    
\end{proposition}
 
We remind that by our convention, we do not distinguish metric spaces and the corresponding distance matrices. Thus, we say that an $n$-tuple of points of $\mathbb{R}^d$ is a realization of an $n\times n$ distance matrix. Furthermore, we  say that a subset of indices $X\subseteq\{1,\ldots,n\}$ is a metric basis of the distance matrix meaning that the corresponding points of the distance space compose a metric basis for it. In this spirit, if $M$ is the distance matrix of a (strongly) $d$-embeddable distance space, then we say that $M$ is a {\em (strongly) $d$-embeddable EDM}.

\paragraph*{Graphs and matrices} We consider simple finite undirected graphs and refer to  \cite{Diestel} for standard graph-theoretic notation. Given a graph $G$, we denote by $V(G)$ and $E(G)$ the sets of vertices and edges, respectively. Throughout the paper we use $n$ and $m$ for $|V(G)|$ and $|E(G)|$, respectively, if the graph is clear from the context. A set of pairwise adjacent vertices $K$ of $G$ is called a \emph{clique} and a set of pairwise non-adjacent vertices is \emph{independent}. 
 A family of $k$ cliques $\mathcal{C}=\{C_1,\ldots, C_k\}$ of a graph $G$ is an \emph{edge clique cover} of $G$ if for every edge $uv\in E(G)$, there is $i\in\{1,\ldots,k\}$ such that $uv$ is \emph{covered} by $C_i$, that is, $u,v\in C_i$.
For a vertex $v\in V(G)$, we use $N_G(v)$ to denote the \emph{open neighborhood} of $v$, that is, the set of vertices adjacent to $v$. 
The \emph{degree} $d_G(v)=|N_G(v)|$, and the \emph{maximum degree} of $G$ is 
$\Delta(G)=\max\{d_G(v)\mid v\in V(G)\}$. For a graph $G$, its \emph{complement} $\overline{G}$ is the graph with the same set of vertices as $G$, and two distinct vertices $u,v\in V(G)$ are adjacent in $\overline{G}$ if and only if they are not adjacent in $G$.

Let $A$ be an $n \times n$ matrix.
The matrix obtained from $A$ by deleting $n-k$ rows and $n-k'$ columns (with $1 \le k,k' \le n$) is called a $k \times k'$ \emph{submatrix} of $A$.
A \emph{principal submatrix} of $A$ is a square submatrix obtained by deleting the same set of row and column indices; that is, if the $i$-th row of $A$ is deleted, then the $i$-th column is also deleted. Let $G$ be a graph with $V(G)=\{1,\ldots,n\}$.
 An $n \times n$ matrix $A = (a_{ij})$ over $\mathbb{R}_{\geq 0}$ with some unspecified elements is said to be a \emph{$G$-partial} matrix if (i)~the entry $a_{ij}$ is defined (or specified) if and only if $ij \in E(G)$, (ii)~$a_{ij} = a_{ji}$ for all $ij \in E(G)$, and $a_{ii}=0$ for all $i\in\{1,\ldots,n\}$. We also say that $G$ is the \emph{underlying} graph of $A$. 
Let $A$ be a $G$-partial matrix and let $d$ be a positive integer. A matrix $D$ is said to be a  \emph{$d$-\edm completion} of $A$ if:
\begin{enumerate}
    \item $D$ is $d$-embeddable, and
    \item $d_{ij} = a_{ij}$ for all $i j \in E(G)$.
\end{enumerate}
See \Cref{Fig:example} for an example.

\begin{figure}[ht]
\centering

% Row 1: G, A, D
\begin{minipage}[t]{0.2\textwidth}
\centering
 \vspace{0.5em}
%\vspace{-1.2em}
\begin{tikzpicture}[scale=1, every node/.style={circle, draw, fill=white, minimum size=2mm}]
    \node (v1) at (0,0) {$1$};  %{$p_1$};
    \node (v2) at (1.5,0) {$2$};  %{$p_2$};
    \node (v3) at (0.75,1.2) {$3$}; %{$p_3$};
    \node (v4) at (0.75,2.2) {$4$}; %{$p_4$};

    \draw (v1) -- (v2);
    \draw (v2) -- (v3);
    \draw (v1) -- (v3);
    \draw (v3) -- (v4);
\end{tikzpicture}
\end{minipage}
\hfill
\begin{minipage}[t]{0.35\textwidth}
\centering
\vspace{0.5em}
\[
A =
\begin{bmatrix}
0 & 1 & 1.25 & * \\
1 & 0 & 1.25 & * \\
1.25 & 1.25 & 0 & 1 \\
* & * & 1 & 0
\end{bmatrix}
\]
\end{minipage}
\hfill
\begin{minipage}[t]{0.35\textwidth}
\centering
\vspace{0.5em}
\[
D =
\begin{bmatrix}
0 & 1 & 1.25 & 4.25 \\
1 & 0 & 1.25 & 4.25 \\
1.25 & 1.25 & 0 & 1 \\
4.25 & 4.25 & 1 & 0
\end{bmatrix}
\]
\end{minipage}

\vspace{1.5em}

\vspace{1.5em}

\caption{Graph $ G$, $G$-partial matrix $ A $, and \edm completion $ D$ of $A$. The distance space  $\Xcal$  defined by $D$ is embeddable into $\mathbb{R}^2$. For example, a realization $\varphi$ of $\Xcal$ would map the elements of $\Xcal$  to points $p_1= (0,0)$, $p_2= (1,0)$, $p_3= (0.5,1)$, and $p_4= (0.5,2)$.}\label{Fig:example}
\end{figure}
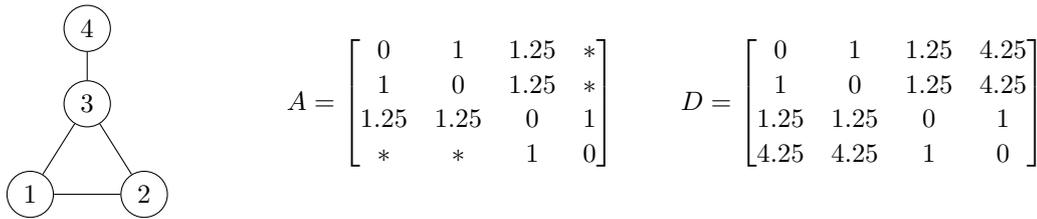

%As we are working with \red{symmetric $G$-partial matrices}, it 
It is convenient for us to combine graph-theoretic and matrix notation. Let  $A=(a_{ij})$ be symmetric $n\times n$ matrix and let $G$ be a graph with $V(G)=[n]$. For a set $X\subseteq [n]$, we use $G[X]$ to denote the subgraph of $G$ induced by $X$ and we use $A[X]$ to denote the principal submatrix indexed by $X$, that is, the submatrix of $A$ composed by the elements $a_{ij}$ with $i,j\in X$. We also write $G-X$ and $A-X$ to denote the graph obtained by the deletion of the vertices of $X$ and the submatrix obtained by the deletion of the rows and columns indexed by $X$, respectively, that is, $G-X=G[V(G)\setminus X]$ and $A-X=A[V(G)\setminus X]$. For a set $X=\{x\}$, we write $G-x$ and $A-x$ instead of $G-\{x\}$ and $A-\{x\}$, respectively.

We conclude this section with \Cref{obs:dense}, which shows that for tractability of \dedmcompshort{}, it is not sufficient to require high density of the specified element in the input matrix (or, equivalently, high density of the underlying graph $G$).

\denseObs*

\begin{proof}
We use the result of Saxe~\cite{saxe1979embeddability} that \dedmcompshort{} is strongly \classNP-complete for $d=1$. More precisely, Saxe~\cite[Theorem~4.1]{Saxe1979} proved that it is \classNP-complete to decide, given a weighted connected graph $G$ with a weight function $w\colon E(G)\rightarrow \{1,2,3,4\}$, whether there is a mapping $\varphi$ of the vertices of $G$ to points of the line such that for every edge $uv\in E(G)$, $|\varphi(u)-\varphi(v)|=w(uv)$. 

Let $G$ be an $n$-vertex connected graph with $V(G)=\{1,\ldots,n\}$, and let $w\colon E(G)\rightarrow \{1,2,3,4\}$ be a weight function. We use the property that the length of an arc of the circle of radius $r$ whose central angle is $\alpha$ equals $r\alpha$, and the length of the chord of the arc is $2r\sin(\alpha/2)$.
Let $\alpha<\frac{2\pi}{4n}=\frac{\pi}{2n}$ and let $r\in\mathbb{R}_{>0}$. We define 
$\dist_{u,v}=2r\sin(\frac{w(uv)\alpha}{2})$ for $uv\in E(G)$. Observe that because $G$ has no cycle of length more than $4n$ and $\alpha<\frac{\pi}{2n}$,
a mapping $\varphi$ of $V(G)$ to the line preserving the lengths of the edges exists if and only if there is a mapping  $\psi$ of $V(G)$ to the circle of radius $r$ such that $\|\psi(u)-\psi(v)\|_2^2=\dist_{u,v}^2$ for every $uv\in E(G)$. We next describe how to select the values $\alpha$ and $r$ to ensure that $r$ and the values of $\dist_{u,v}^2$ are integer.

Notice that $\arcsin(x)<\frac{\pi}{2}x$ if $0<x<1$. We select $\alpha=2\arcsin(\frac{1}{2n})$. Then $\alpha<2\frac{\pi}{2}\frac{1}{2n}=\frac{\pi}{2n}$. 
Then $\sin(\frac{\alpha}{2})=\frac{1}{2n}$. Recall that if $0\leq x\leq\pi/2$ then
$\cos x=\sqrt{1-\sin^2x}$. Then 
$\sin(2x)=2\sin x\sqrt{1-\sin^2x}$, $\sin(3x)=3\sin x-4\sin^3 x$, and 
$\sin(4x)=4\sqrt{1-\sin^2x}(\sin x-2\sin^3x)$. 
We define $r=2^4n^4$. Then because $w(uv)\in\{1,2,3,4\}$ for every $uv\in E(G)$, we obtain that $\dist_{u,v}^2=4r^2\sin^2(\frac{w(uv)\alpha}{2})$ is an integer for each $uv\in E(G)$. Moreover, each $\dist_{u,v}^2=\Oh(n^6)$.  

We construct the following instance $(M,d)$ of \dedmcompshort{} for $d=2$. Set $s=\lceil\frac{n^2}{2\varepsilon}\rceil$ and construct partial $N\times N$ symmetric matrix $M=(m_{ij})$ with zero diagonal elements for $N=n+s$:
\begin{itemize}
\item for distinct $i,j\in\{1,\ldots,n\}$, set $m_{ij}=\dist_{i,j}^2=4r^2\sin^2(\frac{w(uv)\alpha}{2})$ if $ij\in E(G)$, and let $m_{ij}$  be unspecified if $ij\notin E(G)$,
\item for distinct $i,j\in\{n+1,\ldots,n+s\}$, set 
$m_{ij}=0$,
\item for all $i\in\{1,\ldots,n\}$ and $j\in\{n+1,\ldots,n+s\}$, set $m_{ij}=m_{ji}=r^2$.
\end{itemize}

We show that $G$ admits a mapping $\varphi$ of the vertices of $G$ to points of the line such that for every edge $uv\in E(G)$, $|\varphi(u)-\varphi(v)|=w(uv)$ if and only if $(M,2)$ is a yes-instance of \dedmcompshort{}. 

For the forward direction, recall that if there is a mapping $\varphi$ preserving edge lengths then there is a mapping  $\psi$ of $V(G)$ to the circle $C$ of radius $r$ such that $\|\psi(u)-\psi(v)\|_2^2=\dist_{u,v}^2$ for every $uv\in E(G)$. We assume that $C$ has its center in $(0,0)$, and set $p_i=\psi(v_i)$ for $i\in\{1,\ldots,n\}$. Then we construct the completion $M^*=(m_{ij}^*)$ of $M$ by setting $m_{ij}^*=\|p_i-p_j\|_2^2$ for all distinct $i,j\in\{1,\ldots,n\}$ such that $ij\notin E(G)$.
We have that $M^*$ is a \edm because $(p_1,\ldots,p_N)$ where $p_j=(0,0)$ for $j\in \{n+1,\ldots,n+s\}$ is a realization of $M^*$ by the construction. 

For the opposite direction, assume that $M^*$ is a \edm completing $M$. Let $(p_1,\ldots,p_N)$ be an $N$-tuple of the points of $\mathbb{R}^2$ realizing $M^*$. Since $m_{ij}=0$ for distinct $i,j\in\{n+1,\ldots,n+s\}$, the points $p_{i}=p$ for $i\in\{n+1,\ldots,n+s\}$ are the same. Then the points $p_1,\ldots,p_n$ are on the circle $C$ of radius $r$ centered in $p$ as
$m_{ij}=m_{ji}=r^2$ for all $i\in\{1,\ldots,n\}$ and $j\in\{n+1,\ldots,n+s\}$. This means that there is a  mapping  $\psi$ of $V(G)$ to the circle of radius $r$ such that $\|\psi(u)-\psi(v)\|_2^2=\dist_{u,v}^2$ for every $uv\in E(G)$. This implies that there is a mapping $\varphi$ of the vertices of $G$ to points of the line such that for every edge $uv\in E(G)$, $|\varphi(u)-\varphi(v)|=w(uv)$. This completes the correctness proof for the reduction.

As $G$ has at most $\binom{n}{2}$ pairs of nonadjacent vertices, we have that $M$ has at most
$\binom{n}{2}$ unspecified elements. Then the number of unspecified elements is at most
$\binom{n}{2}\leq \frac{n^2}{2}\leq \varepsilon\lceil\frac{n^2}{2\varepsilon}\rceil= \varepsilon s\leq \varepsilon (n+s)=\varepsilon N$. This completes the proof.
\end{proof}

\section{Compressing $G$-partial matrices when $G$ is dense (distance from triviality I and II)}\label{sec:sparse}

This section is devoted to proving  \Cref{thm:compression}, \Cref{thm:Deltacompr} and \Cref{thm:Covercompr}. 

\compressionThm*

\begin{proof}

The main step of the proof is to identify an ``irrelevant'' vertex of $G$ in polynomial time, that is, a vertex whose deletion does not transform a no-instance into a yes-instance. 
%It is trivial to see that a yes-instance cannot be transformed into a no-instance by deleting a vertex of $G$.  
By exhaustively repeating this step of removing an irrelevant vertex we either solve the problem, or construct an equivalent instance on  $(d+ 1)^{\Oh(t^2)}$ rows and columns.

Let $M$ be a $G$-partial matrix, where $G$ excludes a $t$-block pattern. If $t=1$, then $M$ is fully specified (i.e., $G$ is complete) and so, we solve the instance $(M,d)$ using \Cref{prop:realization}. So, we assume that $t>1$. Let $\eta(d,t)=(d+1)t+(t-1)(d+1)^{t+1}+1,$ 
and
 $\rho(d,t)=\binom{2t+\eta(d,t)-2}{2t-1}.$

If $n< \rho(d,t)$, then the desired compression is already achieved and we simply output $(M,d)$. From now on, we assume that $n\geq \rho(d,t)$. Since $\overline{G}$ is $K_{t,t}$-free, it does not contain a clique of size $2t$.
By Ramsey's classic theorem~\cite{ErdosS1935}   $\overline G$ must therefore contain an independent set $X$ of size $\eta(d,t)$. Moreover, such a set  $X$ can be computed in polynomial time. 
By definition, $X$ induces a clique in $G$. Notice that because $M$ is not fully specified, $V(G)\neq X$.

We use \Cref{prop:realization} to check in polynomial time whether $M[X]$ is $d$-embeddable. If not, then  we  report that $(M,d)$ is a no-instance and stop.  For every $v\in V(G)\setminus X$, notice that
$C_v=\{v\}\cup (N_G(v)\cap X)$ is a clique of $G$. So, for every $v\in V(G)\setminus X$ we also check whether $M[C_v]$ is $d$-embeddable, and report that $(M,d)$ is a no-instance if it is not.

We next iteratively construct disjoint sets $X_1,\ldots,X_t\subseteq X$ of size at most $d+1$ as follows. Assuming that $X_0=\emptyset$, we select $X_i$ to be a metric basis of $M[X\setminus\bigcup_{j=0}^{i-1}X_j]$ for $i\in\{1,\ldots,t\}$ obtained using~\Cref{prop:realization}. As $|X_i|\leq d+1$ for $i\in\{1,\ldots,t\}$, we have that $|\bigcup_{i=1}^tX_i|\leq (d+1)t$. Moreover, since $X$ is large enough, the sets $X_1,\dots,X_t$ exist.

Consider $Y=\{y\in V(G)\setminus X\colon X_i\setminus N_G(y)\neq\emptyset\text{ for all }i\in\{1,\ldots,t\}\}$. 
%\end{equation*}
That is, $Y$ consists of all vertices outside $X$ having a non-neighbor in each $X_i$ in $G$. 
%We show that $Y$ has bounded size.

\begin{claim}
%\deferProof
\label{cl:sizeY}
$|Y|\leq (t-1)(d+1)^t$.
\end{claim}
\begin{proof}
%[Proof of \Cref{cl:sizeY}]
Assume for the sake of contradiction that $|Y|> (t-1)(d+1)^t$.
For every $y\in Y$, $X_i\setminus N_G(y)\neq\emptyset$ for each $i\in\{1,\ldots,t\}$. Hence, for each $y\in Y$ and every $i\in\{1,\ldots,t\}$, $y$
has a neighbor $x_i^y\in X_i$ in $\overline{G}$. Notice that because $|X_i|\leq d+1$ for each $i\in\{1,\ldots,t\}$, there are at most $(d+1)^t$ distinct $t$-tuples $(x_1^y,\ldots,x_t^y)$. Thus, since 
$|Y|> (t-1)(d+1)^t$, by the pigeonhole principle, $Y$ contains at least $t$ vertices $y$ associated with the same tuple 
$(x_1^y,\ldots,x_t^y)$. However, this means that $\overline{G}$ contains a $K_{t,t}$ which is a contradiction. This proves the claim.
\end{proof}

For every $y\in Y$, we consider the clique $K_y=N_G(y)\cap X$. Since $K_y\subseteq C_y$ and we have already concluded that $M[C_y]$ is $d$-embeddable, we have that $M[K_y]$ is $d$-embeddable. 
%an EDM in $\mathbb{R}^d$. 
We compute a metric basis $Z_y$ of $M[K_y]$ using \Cref{prop:realization}. Because  
$|Y|\leq (t-1)(d+1)^t$ by \Cref{cl:sizeY} and $|Z_y|\leq d+1$ for $y\in Y$, we have that 
$|\bigcup_{y\in Y}Z_y|\leq (t-1)(d+1)^{t+1}$. Because $|\bigcup_{i=1}^tX_i|\leq (d+1)t$ and $|X|\geq (d+1)t+(t-1)(d+1)^{t+1}+1$, we have that there is $w\in X\setminus\Big(\Big(\bigcup_{i=1}^tX_i\big)\cup\big(\bigcup_{y\in Y}Z_y\Big)\Big)$.

\begin{claim}\label{cl:irrelevant}
The instances $(M,d)$ and $(M-w,d)$ of \dedmcompshort{} are equivalent. 
\end{claim}

\begin{proof}
%[Proof of \Cref{cl:irrelevant}]
It is straightforward that if $(M,d)$ is a yes-instance then the same holds for $(M-w,d)$.

% For the opposite implication.
Let us now argue the converse.   Let $M=(m_{ij})$ and assume without loss of generality that $w=n$.
Assume that $(\hat{M},d)$ is a yes-instance for $\hat{M}=M-w$. Notice that $\hat{M}$ is a $\hat{G}$-partial matrix for $\hat{G}=G-w$. We have that $\hat{M}$ can be completed to an $(n-1)\times (n-1)$ EDM matrix $\hat{M}^*=(m_{ij}^*)$ for $i,j\in\{1,\ldots,n-1\}$. Hence, there is an $(n-1)$-tuple of points $(p_1,\ldots,p_{n-1})$ of $\mathbb{R}^d$ realizing $\hat{M}^*$, that is, $\|p_i-p_j\|_2^2=m_{ij}^*$ for $i,j\in\{1,\ldots,n-1\}$. We will extend $(p_1,\ldots,p_{n-1})$ by adding a point $p_w$ realizing $w$, and construct $M^*=(m_{ij}^*)$ from $\hat{M}^*$ by adding a row and a column such that $M^*$ is a completion of $M$. 

Recall that $X$ is a clique of $G$ and $X_1$ is a metric basis of the EDM $M[X]$. Let $d'=|X_1|-1$. Then the points $p_v$ for $v\in X\setminus \{w\}$ are embedded in the $d'$-dimensional affine subspace $L$ of $\mathbb{R}^d$ defined by the points $p_v$ for $v\in X_1$. Furthermore, because $M[X]$ is an EDM, there is a unique point $p_w\in \mathbb{R}^d$ such that $\|p_w-p_v\|_2^2=m_{vw}$ for $v\in X_1$, by \Cref{prop:basis}. Notice that $p_w\in L$.  
For every $v\in V(G)\setminus \{w\}$, we define $m_{vw}^*=m_{wv}^*=\|p_v-p_w\|_2^2$. We claim that for every $v\in V(G)\setminus \{w\}$ such that $vw\in E(G)$, $m_{vw}^*=m_{vw}$.

We consider three cases depending on where $v$ lies in the partition $X\uplus Y\uplus V(G)\setminus(X\cup Y)$ of $V(G)$. The proof is based on the following observation. 
Assume that $C$ is a clique of $G$ such that $M[C]$ is $d$-embeddable and let $B\subseteq C$ be a metric basis of $M[C]$. Then for any realization $(q_v)_{v\in C}$ of $M[C]$, the $|C|$-tuple of points $(q_v)_{v\in C}$ is the same as $(p_v)_{v\in C}$ up to rigid transformations, by \Cref{prop:basis}. Moreover, it is completely defined by $M$. Also, if $L$ is the $d'$-dimensional affine subspace of $\mathbb{R}^d$ containing $(p_v)_{v\in B}$ for $d'=|B|-1$, then all the points $p_u$ for $u\in C$ are in $L$, and by \Cref{prop:basis}, their coordinates in $L'$ are uniquely defined by the values $m_{vu}$ for $v\in B$.  

%\noindent \textbf{Case~1.} 
\begin{description}

\item[Case 1: $v\in X$.] Then the equality $m_{vw}^*=m_{vw}$ immediately follows from the definition of $p_w$ and the fact that $X_1$ is a metric basis of $M[X]$.

\item[Case 2: $v\in Y$.] Notice that in this case, $w\in K_v$. Recall that $Z_v\subseteq X$ is a metric basis of $M[K_v]$. Let $d''=|Z_v|-1\leq d'$. Since $K_v\subseteq X$, 
the points $p_u$ for $u\in Z_v$ are in $L$. Therefore, there is a $d''$-dimensional affine subspace $L'\subseteq L$ containing the points $p_u$ for $u\in K_v$ including $u=w$. Suppose that $p_v\in L'$. Then $Z_v$ is a metric basis of $M[C_v]$ for the clique $C_v$ of $G$. Since a metric basis defines the positions of the points in the realization in a unique way, we have that $m_{vw}^*=m_{vw}$. Assume that $p_v\notin L'$. Then $Z_v\cup \{v\}$ is a metric basis of $M[C_v]$ by \Cref{prop:matroid}~(iii).  %of size $d''+1$. 
This implies that $m_{vw}^*=m_{vw}$ by the uniqueness of the realization.

\item [Case 3: $v\notin (X\cup Y)$.] 
Then by the definition of $Y$, there is $i\in\{1,\ldots,t\}$ such that $X_i\subseteq N_G(v)$. As $vw\in E(G)$, we have that $C_v'=X_i\cup\{v,w\}$ is a clique of $G$. By the choice of $X_i$, $X_i$ is a metric basis of $M[X\setminus\bigcup_{j=0}^{i-1}X_j]$. Then for $d''=|X_i|-1$, there is a $d''$-dimensional affine subspace $L'\subseteq L$ containing the points $p_u$ for $u\in X_i\cup\{w\}$. If $p_v\in L'$ then $X_i$ is a metric basis of $M[C_v']$ and $m_{vw}^*=m_{vw}$. If $p_v\notin L'$, then $X_i\cup \{v\}$ is a metric basis of $M[C_v']$ by \Cref{prop:matroid}~(iii). Then again we have that $m_{vw}^*=m_{vw}$. This completes the case analysis.
\end{description}

Because $m_{vw}^*=m_{wv}^*=m_{vw}$ for all $v\in V(G)\setminus \{w\}$, we obtain that $M^*$ is a completion of $M$. By construction of $(p_1,\ldots,p_{n-1})$ and $p_w$, we have that $(p_1,\ldots,p_n)$ is a realization of $M^*$ in $\mathbb{R}^d$. Thus, $M^*$ is a completion of $M$ in $\mathbb{R}^d$. This means that $(M,d)$ is a yes-instance of \dedmcompshort{} and completes the proof.
\end{proof}

%By \Cref{cl:irrelevant}, $w$ is irrelevant. Then 
%we 
We set $M:=M-w$ and iterate. In at most $n$ rounds, we either solve the problem or obtain an equivalent instance $(M',d)$, where $M'$ is a $n'\times n'$ submatrix of $M$ with $n'<\binom{2t+\eta(d,t)-2}{2t-1}$. For the running time, note that the construction of the independent set $X$ as well as that of $X_1,\ldots,X_t$ is done in polynomial time by \Cref{prop:realization}. Then $Y$ and the sets $Z_y$ for $y\in Y$ can also be constructed in polynomial time by \Cref{prop:realization}. So, an irrelevant vertex can be found in polynomial time. Because the total number of rounds is at most $n$, the overall running time is polynomial. 
%This concludes the proof.
\end{proof}

As discussed, by combining \Cref{thm:compression} with \Cref{thm:exactAlgoCompletion}, we obtain the following corollary.

\fptKttCor*

%\newpage

In the remaining part of the section, we show that for some other classes of dense graphs $G$, we can obtain a better compression for \dedmcompshort{} than the one obtained by using \Cref{thm:compression} as a black box. The proofs of these results still follow the same lines as the proof of \Cref{thm:compression}. 
First, we consider the case when $\overline{G}$ has bounded maximum degree. Since a graph of maximum degree at most $\Delta$ does not contain $K_{\Delta+1,\Delta+1}$, in the following theorem, we consider a subclass of matrices considered in \Cref{thm:compression}.

\DeltaCompr*

\begin{proof}
%[Proof of \Cref{thm:Deltacompr}]
Let $(M,d)$ be an instance of \dedmcompshort{} such that $M$ is a $G$-partial matrix and the maximum degree of $\overline{G}$ is at most $\Delta$. If $G$ is a complete graph then we solve the problem directly using \Cref{prop:realization}, and if $n\leq (d+1)(\Delta+1)^2$ then we immediately output the instance. Thus, $n>(d+1)(\Delta+1)^2$. 
Because the maximum degree of $\overline{G}$ is at most $\Delta$, $\overline{G}$ has an independent set $X$ of size at least $(d+1)(\Delta+1)+1$ which can be computed greedily. We prove that either $X$ has an irrelevant vertex $w$ or $(M,d)$ is a no-instance. 

Because $X$ is a clique of $G$, $M[X]$ should be an EDM in $\mathbb{R}^d$. We verify this by making use of  \Cref{prop:realization} and return that $(M,d)$ is a no-instance if this is not the case. Also, for every $v\in V(G)\setminus X$, we consider the clique $C_v=\{v\}\cup (N_G(v)\cap X)$. If there is $v\in V(G)\setminus X$ such that $M[C_v]$ is not an  EDM in $\mathbb{R}^d$, then we report that $(M,d)$ is a no-instance and stop. 
From now on, we assume that this is not the case.

We iteratively construct disjoint sets $X_1,\ldots,X_{\Delta+1}\subseteq X$ of size at most $d+1$ as follows. Assuming that $X_0=\emptyset$, we select $X_i$ to be a metric basis of $M[X\setminus\bigcup_{j=0}^{i-1}X_j]$ for $i\in\{1,\ldots,\Delta+1\}$ using \Cref{prop:realization}. Because $|X_i|\leq d+1$ for $i\in\{1,\ldots,\Delta+1\}$, $|\bigcup_{i=1}^{\Delta+1}X_i|\leq (d+1)(\Delta+1)$. Since $|X|\geq (d+1)(\Delta+1)+1$, there is 
$w\in X\setminus\Big(\bigcup_{i=1}^{\Delta+1}X_i\Big)$. We prove that $w$ is irrelevant.

\begin{claim}\label{cl:Deltairr}
The instances $(M,d)$ and $(M-w,d)$ of \dedmcompshort{} are equivalent. 
\end{claim}

The proof of this claim follows along the same lines as the proof of \Cref{cl:irrelevant}, so we have included it after the proof of this theorem.

Since $w$ is irrelevant, we set $M:=M-w$ and iterate. In at most $n$ rounds, we either solve the problem or obtain an equivalent instance $(M',d)$ where $M'$ is a $n'\times n'$  submatrix of $M$ with $n'
\leq (d+1)(\Delta+1)^2$. This concludes the description of the algorithm and its correctness proof.   

The construction of $X$ can be done in polynomial time. Then $X_1,\ldots,X_{\Delta+1}$ are constructed in polynomial time by \Cref{prop:realization}. Thus, we find $w$ in polynomial time. Since the number of rounds is at most $n$, the overall running time is polynomial. This concludes the proof.
\end{proof}

\begin{proof}[Proof of \Cref{cl:Deltairr}.] 
As mentioned, the proof follows the same lines as the proof of \Cref{cl:irrelevant}. It is sufficient to prove that if $(M,d)$ is a yes-instance, then the same holds for $(M-w,d)$. Let $M=(m_{ij})$ and assume without loss of generality that $w=n$. Assume that $(\hat{M},d)$ is a yes-instance for $\hat{M}=M-w$. Notice that $\hat{M}$ is a $\hat{G}$-partial matrix for $\hat{G}=G-w$. We have that $\hat{M}$ can be completed to a $(n-1)\times (n-1)$ EDM matrix $\hat{M}^*=(m_{ij}^*)$ for $i,j\in\{1,\ldots,n-1\}$. Hence, there is an $(n-1)$-tuple of points $(p_1,\ldots,p_{n-1})$ of $\mathbb{R}^d$ realizing $\hat{M}^*$, that is, $\|p_i-p_j\|_2^2=m_{ij}^*$ for $i,j\in\{1,\ldots,n-1\}$. We extend $(p_1,\ldots,p_{n-1})$ by adding a point $p_w$ realizing $w$, and construct $M^*=(m_{ij}^*)$ from $\hat{M}^*$ by adding a row and a column such that $M^*$ is a completion of $M$. 

We use the fact that $X$ is a clique of $G$ and $X_1$ is a metric basis of the EDM $M[X]$. Let $d'=|X_1|-1$. Then the points $p_v$ for $v\in X\setminus \{w\}$ are embedded in the $d'$-dimensional affine subspace $L$ of $\mathbb{R}^d$ defined by the points $p_v$ for $v\in X_1$. Furthermore, because $M[X]$ is an EDM, there is a unique point $p_w\in L$ such that $\|p_w-p_v\|_2^2=m_{vw}$ for $v\in X_1$ by \Cref{prop:basis}.   
For every $v\in V(G)\setminus \{w\}$, we define $m_{vw}^*=m_{wv}^*=\|p_v-p_w\|_2^2$. We show that for every $v\in V(G)\setminus \{w\}$ such that $vw\in E(G)$, $m_{vw}^*=m_{vw}$. We have the following two cases.

%\noindent \textbf{Case~1.}

\begin{description}
\item[Case 1: $v\in X$.] Then the equality $m_{vw}^*=m_{vw}$ follows from the fact that $X_1$ is a metric basis of $M[X]$ and the definition of $p_w$.

\item[Case 2: $v\in V(G)\setminus X$.] Because the maximum degree of $\overline{G}$ is at most $\Delta$, 
there is $i\in\{1,\ldots,\Delta+1\}$ such that $X_i\subseteq N_G(v)$. As $vw\in E(G)$, we have that $C_v'=X_i\cup\{v,w\}$ is a clique of $G$. By the choice of $X_i$, $X_i$ is a metric basis of $M[X\setminus\bigcup_{j=0}^{i-1}X_j]$. Then for $d''=|X_i|-1$, there is a $d''$-dimensional affine subspace $L'\subseteq L$ containing the points $p_u$ for $u\in X_i\cup\{w\}$. If $p_v\in L'$ then $X_i$ is a metric basis of $M[C_v']$ and $m_{vw}^*=m_{vw}$. If $p_v\notin L'$, then $X_i\cup \{v\}$ is a metric basis of $M[C_v']$ by \Cref{prop:matroid}~(iii). Then again we have that $m_{vw}^*=m_{vw}$. %This completes the case analysis.

\end{description}

Since $m_{vw}^*=m_{wv}^*=m_{vw}$ for all $v\in V(G)\setminus \{w\}$, $M^*$ is a completion of $M$, and  $(p_1,\ldots,p_n)$ is a realization of $M^*$ in $\mathbb{R}^d$. Thus, $M^*$ is a completion of $M$ in $\mathbb{R}^d$. This completes the proof.
\end{proof}

Finally in this section, we prove \Cref{thm:Covercompr}.
% which we restate here.
%
\CoverCompr*

\begin{proof} Let $(M,d)$ be an instance of \dedmcompshort{} such that $M$ is a $G$-partial matrix and 
$\mathcal{C}=\{C_1,\ldots,C_k\}$ is an edge clique cover of $G$. If $G$ is a complete graph then we solve the problem directly using \Cref{prop:realization}, and we immediately output the instance if $|V(G)|\leq (d+1)k^2$. For each $i\in\{1,\ldots,k\}$, we check whether $M[C_i]$ is an EDM in $\mathbb{R}^d$ using \Cref{prop:realization}. If this does not hold for some $i\in\{1,\ldots,k\}$, we conclude that $(M,d)$ is a no-instance and stop. 
From now on, we assume that $n>(d+1)k^2$ and for each $i\in\{1,\ldots,k\}$, $M[C_i]$ is an EDM in $\mathbb{R}^d$. We prove that there is an irrelevant vertex.

If $G$ has an isolated vertex $w$ then $w$ is trivially irrelevant because the corresponding point can be mapped in $\mathbb{R}^d$ arbitrarily. Assume that $G$ has no isolated vertices. Then each vertex belongs to some clique in $\mathcal{C}$. Because $|V(G)|> (d+1)k^2$, there is $i\in\{1,\ldots,k\}$ such that 
$|C_i|>(d+1)k$. For every $j\in\{1,\ldots,k\}$ (including $j=i$), we select a metric basis $X_j$ of $M[C_i\cap C_j]$ using \Cref{prop:realization}. Because $|X_j|\leq d+1$ for $j\in\{1,\ldots,k\}$, $|\bigcup_{j=1}^{k}X_j|\leq (d+1)k$. As $|C_i|> (d+1)k$, there is 
$w\in C_i\setminus\Big(\bigcup_{i=1}^{k}X_i\Big)$. We prove that $w$ is irrelevant.

\begin{claim}\label{cl:CCirr}
The instances $(M,d)$ and $(M-w,d)$ of \dedmcompshort{} are equivalent. 
\end{claim}

Again, the proof uses the same arguments as the proof of \Cref{cl:irrelevant}, so we have placed its proof outside this theorem's proof. 
Because we have an irrelevant vertex $w$,  we set $M:=M-w$ and iterate. In at most $n$ rounds, we either solve the problem or obtain an equivalent instance $(M',d)$ where $M'$ is a $n'\times n'$ submatrix of $M$ with $n'\leq (d+1)k^2$. This concludes the description of the algorithm and its correctness proof.   
Because $\mathcal{C}$ is given, $X_1,\ldots,X_k$ can be constructed in polynomial time by \Cref{prop:realization}. Thus, we find $w$ in polynomial time. Since the number of rounds is at most $n$, the overall running time is polynomial. This concludes the proof.
\end{proof}

\begin{proof}[Proof of \Cref{cl:CCirr}.]
We have to prove that if $(M,d)$ is a yes-instance then the same holds for $(M-w,d)$. Let $M=(m_{ij})$ and assume without loss of generality that $w=n$. Assume that $(\hat{M},d)$ is a yes-instance for $\hat{M}=M-w$. Notice that $\hat{M}$ is a $\hat{G}$-partial matrix for $\hat{G}=G-w$. We have that $\hat{M}$ can be completed to a $(n-1)\times (n-1)$ EDM matrix $\hat{M}^*=(m_{ij}^*)$ for $i,j\in\{1,\ldots,n-1\}$. Hence, there is an $(n-1)$-tuple of points $(p_1,\ldots,p_{n-1})$ of $\mathbb{R}^d$ realizing $\hat{M}^*$, that is, $\|p_i-p_j\|_2^2=m_{ij}^*$ for $i,j\in\{1,\ldots,n-1\}$. We extend $(p_1,\ldots,p_{n-1})$ by adding a point $p_w$ realizing $w$, and construct $M^*=(m_{ij}^*)$ from $\hat{M}^*$ by adding a row and a column such that $M^*$ is a completion of $M$. 

Since $C_i$ is a clique of $G$ and $X_i$ is a metric basis of the EDM $M[C_i]$, 
the points $p_v$ for $v\in C_i\setminus \{w\}$ are embedded in the $d'=(|X_i|-1)$-dimensional affine subspace $L$ of $\mathbb{R}^d$ defined by the points $p_v$ for $v\in C_i$. Because $M[C_i]$ is an  EDM, there is a unique point $p_w\in L$ such that $\|p_w-p_v\|_2^2=m_{vw}$ for $v\in X_1$ by \Cref{prop:basis}.   
For every $v\in V(G)\setminus \{w\}$, we define $m_{vw}^*=m_{wv}^*=\|p_v-p_w\|_2^2$. We show that for every $v\in V(G)\setminus \{w\}$ such that $vw\in E(G)$, $m_{vw}^*=m_{vw}$. We have the following two cases.

\begin{description}\item[Case 1: $v\in C_i$.] Then the equality $m_{vw}^*=m_{vw}$ follows from the fact that $X_i$ is a metric basis of $M[C_i]$ and the definition of $p_w$.

\item[Case 2: $v\in V(G)\setminus C_i$.] Then $v\in C_j$ for some $j\in\{1,\ldots,k\}$ such that $j\neq i$. Recall that $X_j$ is a metric basis of $M[C_i\cap C_j]$. Since $v\in C_j\setminus C_i$,
$C_j'=(C_i\cap C_j)\cup\{v,w\}$ is a clique of $G$. By the choice of $X_j$, for $d''=|X_i|-1$, there is a $d''$-dimensional affine subspace $L'\subseteq L$ containing the points $p_u$ for $u\in X_j\cup\{w\}$. If $p_v\in L'$ then $X_j$ is a metric basis of $M[C_j']$ and $m_{vw}^*=m_{vw}$. If $p_v\notin L'$, then $X_j\cup \{v\}$ is a metric basis of $M[C_j']$ by \Cref{prop:matroid}~(iii). Then again we have that $m_{vw}^*=m_{vw}$. 
\end{description}
As $m_{vw}^*=m_{wv}^*=m_{vw}$ for all $v\in V(G)\setminus \{w\}$, $M^*$ is a completion of $M$, and  $(p_1,\ldots,p_n)$ is a realization of $M^*$ in $\mathbb{R}^d$. Thus, $M^*$ is a completion of $M$ in $\mathbb{R}^d$. This completes the proof.
\end{proof}

Notice that \Cref{thm:Covercompr} assumes that an edge clique cover is given. However, an edge clique cover of size at most $k$ can be found in \classFPT in $k$ time (if exists) by the results of Gramm et al.~\cite{GrammGHN08}.

\begin{proposition}[\cite{GrammGHN08}]\label{prop:ECC}
Given a graph $G$ and a positive integer $k$, it can be decided in $2^{2^{\Oh(k)}}n^{\Oh(1)}$ time whether $G$ admits an edge clique cover of size at most $k$. Furthermore, an edge clique cover can be found in the same time if it exists.     
\end{proposition}

Then we obtain the following corollary.

\fptECC*

\begin{proof}
We first use the algorithm of Gramm et al.~\cite{GrammGHN08} to compute  an edge clique cover of size~$k$
for the underlying graph (or conclude that one does not exist). This takes time $2^{2^{\bigoh(k)}}\cdot n^{\bigoh(1)}$.  Following this, we run the compression in 
\Cref{thm:Covercompr} followed by the algorithm of \Cref{thm:exactAlgoCompletion}. 
\end{proof}

\section{The case of bounded fill-in (distance from triviality III)}\label{sec:fill-in}

In this section, we present our algorithm for {\dedmcompshort} for $G$-partial matrices when $G$ is almost chordal, i.e., it has a small fill-in. Recall that the minimum fill-in of a graph $G$ is the size of a smallest set of non-edges whose addition to $G$ makes it chordal.

\XPFillIn*

Towards the proof of this theorem, let us prepare as follows.
%by recalling some results from the literature. 

\begin{proposition}[{\rm \cite{Bakonyi1995}}]\label{prop:chordalTheorem}
Every $G$-partial matrix $M$ where $G$ is chordal
can be completed to an EDM in ${\mathbb R}^{d}$  if every fully specified principal submatrix of $M$ is $d$-embeddable. 
\end{proposition}

\begin{definition}[\cite{BasuPR06}]\label{def:polynomialFormulas}
{\em   Let $R$ be a real closed field and $\Pcal\subset R[X_1,\dots,X_t]$ be a finite set of polynomials. 
%$\Pcal$-{\em Atoms} are polynomial equations and  inequations created from $\Pcal$. Formally, 
A {\em $\Pcal$-atom} is one of $P=0$,$P\neq 0$, $P\ge 0$, $P\le 0$, where $P$ is a polynomial in $\Pcal$ and a {\em quantifier-free $\Pcal$-formula} is a formula constructed from $\Pcal$-atoms together with the logical connectives $\wedge$, $\vee$ and $\neg$. The $R$-{\em realization} of a formula $F$ with free variables $x_1,\dots,x_r$  is a mapping $\nu:\{x_1,\dots,x_r\}\to R$ so that the sentence resulting from $F$ (denoted by $F[\nu]$) by instantiating each free variable $x$ with $\nu(x)$ is true. 
}
\end{definition}

\begin{proposition}[Theorem 13.13, \cite{BasuPR06}]\label{prop:basuPolynomialModelFindingTheorem}
Let $s,\ell\in {\mathbb N}$. 
Let $(\exists X_1)\dots  (\exists X_t) F(X_1,\dots,X_t),$
    be a sentence, where $F(X_1,\dots,X_t)$ is a quantifier free $\Pcal$-formula where $\Pcal\subset R[X_1,\dots,X_t]$ is a set of at most $s$ polynomials each of degree at most $\ell$. There exists an algorithm to decide the truth of the sentence with complexity\footnote{The measure of complexity here is the number of arithmetic operations. Recall that we use the real RAM model.}
%    \footnote{The measure of complexity here is the 
    $s^{t+1}\cdot \ell^{\Oh(t)}$ in $D$ where $D$ is the ring generated by the coefficients of the polynomials in $\Pcal_{\ell}$. 
\end{proposition}

\begin{definition}\label{def:hatMatrix}
{\em Let $M$ be a $G$-partial matrix where $V(G)=[n]$ and let $Z$ denote the set of pairs in $[n]\choose 2$ such that every pair in $Z$ indexes a pair of  unspecified entries of $M$, that is, for every $\{i,j\}\in Z$, $M_{ij}$ and $M_{ji}$ are both unspecified. Define the matrix $\hat M$ as follows.   If $\{i,j\}\notin Z$, then define $\hat M_{ij}=M_{ij}$, otherwise define $\hat M_{ij}=z_{ij}$ where $z_{ij}$ is an indeterminate. 
 }
 \end{definition}

We have the following adaptation of the notion of $Z$-augmented Cayley-Menger determinant from \cite{BentertFGRS25}.  
\begin{definition}\label{def:augmentedCayleyMenger}
{\em Let $M$ and $Z$ be as described in \Cref{def:hatMatrix}.
Let $I=\{x_0,\dots,x_r\}\subseteq [n]$.
The {\em $Z$-Augmented Cayley-Menger determinant indexed by $I$} is obtained from the Cayley-Menger determinant by replacing each $\dist^{2}_{x_i,x_j}$ with $\hat M_{x_i,x_j}$.

That is:

\[
CM_Z ( x_0,x_1, \dots, x_r) =  \det\left( \begin{matrix}
0 & 1 & 1 & 1 & \dots & 1 \\
1 & 0 & \hat M_{x_0x_1} & \hat M_{x_0x_2} & \dots & \hat M _{x_0x_r} \\
1 & \hat M_{x_0x_1} & 0 & \hat M_{x_1x_2} & \dots & \hat M_{x_1x_r} \\
1 & \hat M_{x_0x_2} & \hat M_{x_1x_2} & 0 & \dots & \hat M_{x_2x_r} \\
\vdots & \vdots & \vdots & \vdots & \ddots & \vdots \\
1 & \hat M_{x_0x_r} & \hat M_{x_1x_r} & \hat M_{x_2x_r} & \dots & 0 \\
\end{matrix} \right)
\]

}
\end{definition}

\begin{lemma}
%\deferProof
\label{lem:structureOfAugmentedCM}
Let $M,Z$ and $I=\{x_0,\dots,x_r\}$ be as described in \Cref{def:augmentedCayleyMenger}.  Then, $CM_Z ( x_0,x_1, \dots, x_r)$ is a multi-variate polynomial with real coefficients, over the set  $\{z_{ij}\mid \{i,j\}\in Z\}$ of indeterminates  where each monomial has degree at most ${\rm min}(|I|,2|Z|)$. 
\end{lemma}

	\begin{proof}
	The lemma follows from the  fact that the determinant is computed for a matrix with the following properties. It is a $(|I|+1) \times (|I|+1)$ matrix where the first row and column contain only constants, there are exactly $2|Z|$ distinct variables in total and each variable appears at a unique position in the matrix. 
\end{proof}

\begin{lemma}
\label{lem:computeCMpolynomial}
There is an algorithm that, given $M, I, Z$ as described in \Cref{def:augmentedCayleyMenger}, runs in time $|I|^{\bigoh(|I|)}\cdot n^{\bigoh(1)}$ and produces the polynomial representing the $Z$-Augmented Cayley-Menger determinant indexed by $I$. 
	\end{lemma}

\begin{proof}
Constructing $\hat M$ given $M,I,Z$ is straightforward and clearly a polynomial-time procedure. From this, we construct $\hat M[I]$, after which the algorithm only has to evaluate the determinant of an $|I|+1$-dimensional matrix, which can be done by brute force in the specified time. 
\end{proof}

\begin{lemma}\label{lem:everyMetricBasisWorks}
	Consider an $n\times n$ EDM $M$ and $Y\subseteq [n]$. Let $B\subseteq [n]$ be a metric basis of $M-Y$. Then, there is a metric basis of $M$ that contains $B$. 
\end{lemma}

\begin{proof}
If $B$ is a metric basis of $M$, then we are done, so suppose not.  
Let $B_{\text{opt}}$ be a metric basis of $M$. By our assumption that $B$ is not a metric basis of $M$   and \Cref{prop:matroid} (iv), we have that $|B_{\text{opt}}|>|B|$. Moreover,  by \Cref{prop:matroid} (iii), there exists a subset $B_{\text{ext}} \subseteq B_{\text{opt}}$ with size $|B_{\text{opt}}| - |B|$ such that
$   B' = B \cup B_{\text{ext}}$
    is an independent set. By \Cref{prop:matroid} (iv), it follows that $B'$ is a metric basis of $M$, completing the proof of the lemma.
\end{proof}

\XPFillIn*

\begin{proof}

	Let the instance $(M,d)$ be given, where $M$ is $G$-partial and moreover, let $X$ be a set of at most $k$ non-edges of $G$ such that  $G'=G+X$ is a chordal graph.  It is straightforward to construct $G$
 from $M$ and once we have $G$,  we 
% invoke \Cref{prop:computeFillIn} to
  compute $X$ (if it exists) in time $2^{o(k)}n^{\bigoh(1)}$ using the algorithm of Fomin and Villanger \cite{FominV13}. 
 In the rest of the proof, we will only refer to the graph $G'$, so we can refer to the pairs in $X$ as edges.   Let $\Zcal$ denote a set of $2k$ indeterminates, two per edge in $X$.  Formally, for every edge $uv\in X$, $\Zcal$ contains the indeterminates $z_{uv}$ and $z_{vu}$. 
% $\Zcal=\{z_{uv}\mid uv\in X\}$.
%
Let $C_1,\dots,C_t$ be the maximal cliques in $G'$.
 It is well-known that $t\leq n$ and these cliques can be computed in polynomial time \cite{Golumbic80}.

\begin{enumerate}
\item Do the following for each $i\in [t]$:
\begin{enumerate}
\item  Define $\Zcal_i=\{z_{uv}\mid u,v\in V(C_i)\}$. That is, $\Zcal_i$ comprises those indeterminates corresponding to those edges in $X$ that have both endpoints in $V(C_i)$.  

	\item Define $V_X(C_i)=\{u\mid \exists v:uv\in X \wedge \{u,v\}\subseteq V(C_i)\}$. 
%	\item 
That is, $V_X(C_i)$ is the set of endpoints of edges in $X$ that are contained in $C_i$.

	\item Define $B_i$ to be an arbitrary metric basis of $M_i=M[V(C_i)\setminus V_X(C_i)]$. 
	Since by definition, $M_i$ is fully specified, $B_i$ is well-defined.
	If $|B_i|>d+1$, then stop and conclude that the input is a no-instance.
	\item For every $Y\subseteq V_X(C_i)$ of size at most $d+1-|B_i|$, do the following:
	 \begin{enumerate}\item Define $B^{Y}_i=B_i\cup Y=\{x_0,\dots,x_{r}\}$. 
	\item  Use \Cref{lem:computeCMpolynomial} to construct the following set $P^{Y}_i$ of atoms:
	\begin{enumerate}
		\item   $(-1)^{j+1}CM_{Z_{i}}( x_0,x_1, \dots, x_j)>0$, where $1\leq j\leq r$. 
	\item $CM_{Z_{i}}(x_0,x_1, \dots, x_r,x)=0$ for every $x\in V(C_i)\setminus B^{Y}_i$. 
	\item $CM_{Z_{i}}(x_0,x_1, \dots, x_r,x,y)=0$ for every $x,y\in V(C_i)\setminus B^{Y}_i$.
	\item $z_{uv}\geq 0$ for each indeterminate   $z_{uv}\in Z_i$. 
	\item $z_{uv}-z_{vu}=0$ for each pair $z_{uv},z_{vu}$ of indeterminates in $Z_i$. 
	\end{enumerate}
	\item Define $\phi^{Y}_i$ to be the quantifier-free formula defined as the conjunction of the atoms in $P^{Y}_i$. 
%	\item We argue that the sentence  defined by prepending $\phi^{Y}_i$ with existential quantifications for every indeterminate [TODODEFINE] is true if and only if there is an assignment to the distances between the appropriate pairs that ensures the $d$-embeddability of $C_i$ with $B^{Y}_i$ as a basis. 
	
\end{enumerate}
\item Define $\Psi_i$ to be the quantifier-free formula defined as the disjunction of the formulas $\phi^{Y}_i$ taken over all $Y\subseteq V_X(C_i)$ of size at most $d+1-|B_i|$.
	\end{enumerate}

\item Define the quantifier-free formula $\Gamma=\bigwedge_{i\in [t]}\Psi_i$.

\item  Run the algorithm of  \Cref{prop:basuPolynomialModelFindingTheorem} on the sentence obtained by existentially quantifying the indeterminates in $\Zcal$  and prepending them to $\Gamma$.
 Return the same answer as that produced by this invocation. 
	\end{enumerate}

\begin{claim}\label{clm:correctnessOfFillInAlgorithms}
$(M,d)$ is a yes-instance of {\dedmcompshort} if and only if this algorithm returns {\sf Yes}. 
\end{claim}

	\begin{proof}
We prove both directions of the equivalence. 

%\textbf{Forward direction:} 
Suppose $M$ has a completion to an EDM $M''$. For each edge $uv \in X$, define $z_{uv} = M''_{uv}$.  
 Recall that $M''_{uv}=M''_{vu}\geq 0$, so we have that $z_{uv}=z_{vu}\geq 0$. Let $\nu \colon \mathcal{Z} \to \mathbb{R}$ be the map satisfying $\nu(z_{uv}) = M''_{uv}$ for all $z_{uv} \in \mathcal{Z}$. To show $\Gamma[\nu]$ holds, it suffices to verify $\Psi_i[\nu]$ for every $i \in [t]$, so fix an $i \in [t]$. By the definition of $\Psi_i$, we need to argue the existence of a subset $Y \subseteq V_X(C_i)$ with $|Y| \leq d + 1 - |B_i|$ such that $\phi^Y_i[\nu]$ holds. By \Cref{lem:everyMetricBasisWorks}, there exists a metric basis $B'_i$ of $M''[V(C_i)]$ containing $B_i$. Since $M''[V(C_i)]$ is $d$-embeddable, $|B'_i| \leq d + 1$, so $Y = B'_i \setminus B_i$ satisfies $|Y| \leq d + 1 - |B_i|$. We have already observed that atoms in $P^{Y}_i$ of the type (D) and (E) are already satisfied, so it remains to consider types (A)--(C). However, the fact that these are satisfied
% 
%  The truth of $\phi^Y_i[\nu]$ 
  follows from \Cref{thm:Blumental}, as $M''[V(C_i)]$ is an EDM with metric basis $B'_i$.

%\textbf{Converse direction:} 
Conversely, let $\nu \colon \mathcal{Z} \to \mathbb{R}$ such that $\Gamma[\nu]$ is true. Construct the matrix $\wtilde{M}$ by setting $\wtilde{M}_{ij} = M_{ij}$ if $z_{ij} \notin \mathcal{Z}$, and $\wtilde{M}_{ij} = \nu(z_{ij})$ otherwise. By construction, $\wtilde{M}$ is a $G'$-partial matrix. By \Cref{prop:chordalTheorem}, we need only verify that every fully specified principal submatrix of $\wtilde{M}$ is $d$-embeddable. Each such submatrix is contained in $\wtilde{M}[V(C_i)]$ for some $i \in [t]$, so fix an $i \in [t]$. Since $\Psi_i[\nu]$ is true, there exists $Y \subseteq V_X(C_i)$ with $|Y| \leq d + 1 - |B_i|$ such that $\phi^Y_i[\nu]$ is true. Fix a smallest such $Y$.  Let $B'_i = B_i \cup Y$. By \Cref{thm:Blumental}, $B'_i$ is a metric basis of $\wtilde{M}[V(C_i)]$, showing that $\wtilde{M}[V(C_i)]$ $d$-embeddable.

This completes the proof of the claim. 
\end{proof}

\begin{claim}\label{clm:correctSetOfPolynomials}
For each $i\in [t]$ and $Y\subseteq V_X(C_i)$, 
the atoms in $P^{Y}_i$ are $\Pcal^{Y}_i$-atoms, where $\Pcal^{Y}_i$ is a set of $\bigoh(n^{2})$ polynomials each of degree at most ${\rm min}(d+3,2k)$. 
\end{claim}

	\begin{proof}
Fix $i \in [t]$ and let $Y \subseteq V_X(C_i)$ 
be as described in the statement of the claim.  Consider the atoms in $P^Y_i$, which by definition fall into types (A)--(E). 
  By \Cref{lem:structureOfAugmentedCM}, each augmented CM determinant appearing in atoms of Type (A)-(C) is a polynomial of degree at most ${\rm min}(d+3,2k)$. Moreover, the number of atoms of Type (A)-(C) is dominated by Type (C), of which there are $\bigoh(n^{2})$ since we have $\bigoh(n^{2})$ choices of $x$ and $y$. Moreover types (D) and (E) collectively contain at most $3k$ atoms with degree-1 polynomials.
Finally, recall that $k$ is at most the number of edges in the $n$-vertex graph $G'$, so $k=\bigoh(n^{2})$. This completes the proof of the claim.
\end{proof}

\begin{claim}\label{clm:correctnessOfFillInRunTime}
The algorithm runs in time $d^{\bigoh(kd)}\cdot 2^{\bigoh(k^{2})}\cdot n^{\bigoh(k)}$. 
\end{claim}

\begin{proof} Steps 1(a), 1(b) clearly run in polynomial time. Step 1(c) requires polynomial time (see \Cref{prop:realization}) to compute a metric basis of $M_i$. Step 1(d) processes at most $4^k$ subsets $Y$, computing $\mathcal{O}(n^2)$ atoms in $P^Y_i$ per subset (each atom is computed in $d^{\mathcal{O}(d)}n^{\bigoh(1)}$ time by \Cref{lem:computeCMpolynomial}).

The total construction time for $\Gamma$ in Step 2 is therefore upper bounded by:
$$\sum_{\substack{i \in [t] \\ Y \in \binom{V_X(C_i)}{\leq d+1 - |B_i|}}} \mathcal{O}\left(|P^Y_i|\right)= \mathcal{O}(\rho(d,n,k)),$$
%
% $\sum_{i,Y} \mathcal{O}(|P^Y_i|) $, 
 where $\rho(d,n,k)$ denotes an upper bound on the total number of atoms that appear in $\Gamma$.
 Plugging in \Cref{clm:correctSetOfPolynomials} 
 to bound the size of each $P^{Y}_i$ 
 gives us $\rho(d,n,k) = d^{c_1 d} \cdot 2^k \cdot n^{c_2}$ for some constants $c_1, c_2 \geq 1$.

 All polynomials appearing in atoms that make up $\Gamma$ have degree $\leq \min(d+3, 2k)$ by \Cref{clm:correctSetOfPolynomials}. So, when Step 3 invokes the algorithm of \Cref{prop:basuPolynomialModelFindingTheorem} with at most $2k$ quantified variables, this yields a total runtime upper bounded by $\rho(d,n,k)^{2k+2} \cdot (2k)^{\mathcal{O}(k)}$, which is upper bounded by the claimed running time.
\end{proof}

This completes the proof of the theorem.
	\end{proof}

	Since every $n$-vertex graph has minimum fill-in at most ${n\choose 2}$ and every $n\times n$ EDM can be embedded into $\mathbb{R}^{n}$, it naively follows from  \Cref{thm:XPAlgorithmbyFillIn} by plugging in $k\leq n^{2}$ and $d\leq n$ that {\dedmcompshort} can be solved in time  $2^{\bigoh(n^{4})}$. However, a more careful argument using the approach from \cite{BentertFGRS25} leads to the better runtime stated in \Cref{thm:exactAlgoCompletion}.

\exactAlgoCompletion*

\begin{proof}
The key insight in this proof compared to the proof of \Cref{thm:XPAlgorithmbyFillIn} is that we can simply ``guess'' the metric basis of the hypothetical completed matrix, construct polynomial equations corresponding to the guess and then use \Cref{prop:basuPolynomialModelFindingTheorem}.   The details follow. 

 Define $Z=\{z_{uv}\mid uv\notin E(G)\}$. That is, $Z$ comprises those indeterminates corresponding to the non-edges of $G$. 
For every $Y=\{x_0,\dots,x_{r}\}\subseteq V(G)$ of size  $r\leq d+1$, use \Cref{lem:computeCMpolynomial} to construct the following set $P^{Y}$ of atoms:

\begin{enumerate}
		\item   $(-1)^{j+1}CM_{Z}( x_0,x_1, \dots, x_j)>0$, where $1\leq j\leq r$. 
	\item $CM_{Z}(x_0,x_1, \dots, x_r,x)=0$ for every $x\in V(G)\setminus Y$. 
	\item $CM_{Z}(x_0,x_1, \dots, x_r,x,y)=0$ for every $x,y\in V(G)\setminus Y$.
	\item $z_{uv}\geq 0$ for each indeterminate   $z_{uv}\in Z$. 
	\item $z_{uv}-z_{vu}=0$ for each pair $z_{uv},z_{vu}$ of indeterminates in $Z$. 
	\end{enumerate}
	
	Define $\phi^{Y}$ to be the quantifier-free formula defined as the conjunction of the atoms in $P^{Y}$. 

For each $Y\subseteq V(G)$ of size at most $d+1$, 
run the algorithm of  \Cref{prop:basuPolynomialModelFindingTheorem} on the sentence obtained by existentially quantifying the indeterminates in $Z$  and prepending them to $\phi^{Y}$.
Return \textsc{Yes} if and only if for at least one choice of  $Y$ this invocation on $\exists Z\cdot \phi^{Y}$ returned \textsc{Yes}. The correctness is straightforward since we exhaustively consider all possible metric bases for the hypothetical completed EDM. The claimed running time follows from the fact that each invocation to \Cref{prop:basuPolynomialModelFindingTheorem} is done on a system defined over $n^{\bigoh(1)}$ atoms each of degree at most $|Z|\leq n^{2}$ and there are at most $2^{n}$ such invocations in total. 
 \end{proof}

\section{Concluding remarks and open questions}\label{sec:conclusion}

By \Cref{thm:XPAlgorithmbyFillIn}, the \dedmcompshort{} problem can be solved in polynomial time for every fixed value of $d$ and $k$. 
%So, it is in the class \classXP{} when parameterized by the size of the fill-in \( k \) of the underlying graph and the dimension \( d \). 
Moreover, for fixed values of $k$, our algorithm is \classFPT in $d$. 
It is therefore a natural question whether the problem is in \classFPT{} with respect to {\em both} $d$ and $k$.  

Besides this, our work raises a series of questions asking which other structural parameters of the graph \( G \) yield \classXP{} or \classFPT{} algorithms.

One of the most well-studied parameters in graph algorithms is the \emph{treewidth} of the graph.  
However, a classic result by Saxe~\cite{saxe1979embeddability} shows that the \dedmcompshort{} problem is NP-hard for \( d = 1 \) even when \( G \) is a cycle, which is a graph of treewidth at most two.  
Therefore, parameterized by treewidth and dimension, the problem is \classParaNPComplete.  
The same hardness result applies when parameterizing by the size of the \emph{feedback vertex set} as a cycle has a feedback vertex set of size one.

A more interesting situation arises when parameterizing by the size of the \emph{minimum vertex cover} of \( G \).  
Observe that if a vertex cover of \( G \) has size at most \( k \), then the fill-in of \( G \) is at most \( \binom{k}{2} \).  
Indeed, one can make \( G \) chordal by adding all missing edges between the vertices in the vertex cover.
Another useful observation is that if the vertex cover of \( G \) has size at most \( k \), then we can assume without loss of generality that the embedding dimension \( d \leq k + 1 \).  
This is because once the vertices in the vertex cover are embedded, the remaining vertices form an independent set, and all of them can be placed by making use of at most one additional dimension.
Therefore, by \Cref{thm:XPAlgorithmbyFillIn}, the \dedmcompshort{} problem is in \classXP{} when parameterized by the size of the minimum vertex cover of \( G \).  
It is therefore a natural open question whether the problem is in \classFPT{} or \classW{1}\nobreakdash-hard under this parameterization.

Another parameter of \( G \) whose influence on the complexity of \dedmcompshort{} is not yet understood is   \emph{tree-depth} \cite{NesetrilOdM12}.  
A first step towards this goal would be understanding the parameterized complexity of \dedmcompshort{} 
when parameterized by both the tree-depth and \( d \).

\bibliography{book_pc}

\end{document}